\documentclass[11pt,hidelinks]{article}

\usepackage{orcidlink}

\title{The torus plateau for the high-dimensional Ising model}

\author{
Yucheng Liu\,\orcidlink{0000-0002-1917-8330}\thanks{Department of Mathematics,
	University of British Columbia,
	Vancouver, BC, Canada V6T 1Z2.
	Liu: \href{mailto:yliu135@math.ubc.ca}{yliu135@math.ubc.ca}.
	Slade: \href{mailto:slade@math.ubc.ca}{slade@math.ubc.ca}.
	}
\and
Romain Panis\,\orcidlink{0009-0001-4604-8398}\thanks{Institut Camille Jordan, 43 boulevard du 11 novembre 1918, 69622 Villeurbanne, France. \href{mail:panis@math.univ-lyon1.fr}{panis@math.univ-lyon1.fr}.}
\and
Gordon Slade\,\orcidlink{0000-0001-9389-9497}$^*$
}
\date{\vspace{-5ex}} 

\usepackage{amsmath, amssymb, amscd, amsthm, amsfonts}
\usepackage{graphicx} 
\usepackage{hyperref} 

\usepackage{xcolor}
\usepackage{dsfont, bm}
\usepackage[title]{appendix}
\usepackage{comment}
\usepackage{enumerate}
\usepackage{enumitem}
\usepackage{cite}

\usepackage[textwidth=480pt,textheight=650pt,centering]{geometry} 

\theoremstyle{plain}
\newtheorem{theorem}{Theorem}[section]
\newtheorem{lemma}[theorem]{Lemma}

\newtheorem{proposition}[theorem]{Proposition}

\newtheorem{definition}[theorem]{Definition}

\theoremstyle{definition}
\newtheorem{remark}[theorem]{Remark}

\usepackage{amsfonts}
\usepackage{amsthm}
\usepackage{amsmath}
\usepackage{amssymb}
\usepackage{amscd}
\usepackage{mathrsfs}
\usepackage[mathscr]{eucal}
\usepackage{graphicx}
\usepackage{epstopdf}
\usepackage{epic}
\usepackage{xcolor}
\usepackage{float}

\usepackage[utf8]{inputenc}
\usepackage[T1]{fontenc}
\usepackage{lmodern}
\usepackage{dsfont}

\theoremstyle{plain}

\newtheorem{Lem}[theorem]{Lemma}
\newtheorem{Coro}[theorem]{Corollary}

\newtheorem{Conj}[theorem]{Conjecture}

\theoremstyle{definition}

\newcommand\n{\mathbf{n}}
\newcommand\m{\mathbf{m}}
\renewcommand\k{\mathbf{k}}
\renewcommand\l{\boldsymbol{\ell}}

\newcommand\sn{\partial\mathbf{n}}
\newcommand\sm{\partial\mathbf{m}}

\newcommand\bfC{\mathbf{C}}

\numberwithin{equation}{section}

\newcommand{\Z}{\mathbb{Z}}

\newcommand{\R}{\mathbb{R}}

\newcommand{\E}{\mathbb{E}}

\renewcommand{\P}{\mathbb{P}}
\newcommand{\T}{\mathbb{T}}

\newcommand{\Acal}{\mathcal{A}}
\newcommand{\Ccal}{\mathcal{C}}
\newcommand{\Ecal}{\mathcal{E}}

\newcommand{\Gcal}{\mathcal{G}}
\newcommand{\Hcal}{\mathcal{H}}
\newcommand{\Kcal}{\mathcal{K}}

\newcommand{\Mcal}{\mathcal{M}}
\newcommand{\Ncal}{\mathcal{N}}

\newcommand{\del}{\partial}

\newcommand{\inv}{^{-1}}
\renewcommand{\(}{\left(}
\renewcommand{\)}{\right)}
\newcommand{\half}{\frac{1}{2}}

\newcommand{\1}{\mathds{1}}

\newcommand{\nl}{\nonumber \\}

\providecommand{\abs}[1]{\lvert#1\rvert}
\providecommand{\norm}[1]{\lVert#1\rVert}

\newcommand{\taub}{\tau_\beta}
\newcommand{\chib}{{\chi(\beta)}}
\newcommand{\connect}{\xleftrightarrow}

\newcommand{\bubble}{{\sf B}}

\newcommand{\nnb}{\nonumber\\}
\newcommand{\veee}[1]{|\!|\!|#1|\!|\!|}
\newcommand{\xvee}{\veee{x}}
\providecommand{\nnnorm}[1]{\veee {#1} }

\newcommand{\Tsum}{T}
\newcommand{\Tbr}{\Tsum_\beta^{(r)}}

\DeclareMathAlphabet{\mathpzc}{OT1}{pzc}{m}{it}

\newcommand{\Tr}{{\T_r}}
\newcommand{\Lambdar}{{\Lambda_r}}
\newcommand{\LambdaR}{{\Lambda_R}}
\newcommand{\Zd}{{\Z^d}}

\usepackage{mathtools} 

\begin{document}
\maketitle

\begin{abstract}
We consider the Ising model on a $d$-dimensional discrete torus of volume $r^d$,
in dimensions $d>4$ and for large $r$, in the vicinity of the infinite-volume critical
point $\beta_c$.
We prove that for $\beta=\beta_c- {\rm const}\, r^{-d/2}$ (with a suitable
constant) the susceptibility is bounded above and below by multiples of $r^{d/2}$. Additionally, again for $\beta=\beta_c- {\rm const}\, r^{-d/2}$,
the two-point function has a ``plateau'': it decays
like $|x|^{-(d-2)}$ when $|x|$ is small relative to the volume, but
for larger $|x|$, it levels off to
a constant value of order $r^{-d/2}$.
We also prove that
at $\beta=\beta_c- {\rm const}\, r^{-d/2}$ the renormalised
coupling constant is nonzero, which implies a non-Gaussian limit for the average
spin.  The proof relies on near-critical estimates for the
infinite-volume two-point function obtained recently by Duminil-Copin and Panis,
and builds upon a strategy proposed by Papathanakos.
The random current representation of
the Ising model plays a central role in our analysis.
\end{abstract}

%

\section{Introduction and main results}

\subsection{Introduction}

The Ising model has been studied for a century as a fundamental example of
a phase transition in statistical mechanics.  We are interested here in
the finite-size scaling of the Ising model on a finite box in
the Euclidean lattice $\Z^d$
in dimensions $d>4$, particularly in the case of periodic boundary conditions
which make the box into a torus.

For $d \ge 1$ and an integer $r \ge 3$, we write $\Lambda_r$ for the discrete
box $[-r/2,r/2)^d \cap \Z^d$ of volume $r^d$, and write
$\T_r = (\Z/r\Z)^d$ for the $d$-dimensional discrete torus of period $r$.
We study large $r$, and seek estimates that are uniform in large $r$.
For convenience, we often identify a point $x\in\T_r$ with its representative in
$\Lambda_r$.

\smallskip\noindent
\emph{Notation.}
We write $f\lesssim g$ to denote the existence of a constant
$C>0$ such that $f \leq C g$,
write $f \asymp g$ to mean $f \lesssim g$ and $g \lesssim f$,
and write $f \sim g$ to denote $\lim f/g=1$.
To avoid division by zero, for $x \in \R^d$ we write
$\xvee = 1\vee \|x\|_\infty$.

\smallskip
Let $\Gcal=(V,E)$ be a finite graph with vertex set $V$ and edge set $E$.
The Ising model on $\Gcal$ is a family of probability measures on spin configurations $\sigma : V \to \{-1,+1\}$,
defined using the Hamiltonian
\begin{equation} \label{eq:Hamiltonian}
    H^\Gcal (\sigma) = -\sum_{xy \in E} \sigma_x \sigma_y .
\end{equation}
The expectation of a function $F$ of spin configurations,
at inverse temperature $\beta >0$,
is defined by
\begin{equation} \label{eq:def_Ising}
    \langle F  \rangle_\beta^\Gcal
    =
    \frac{1}{Z^\Gcal_\beta} \sum_{\sigma} F(\sigma) e^{-\beta H^\Gcal(\sigma)} ,
\end{equation}
where $Z^\Gcal_\beta = \sum_{\sigma} e^{-\beta H^\Gcal(\sigma)}$ is the partition function.
Free boundary conditions (FBC) on the box correspond to $\Gcal=\Lambda_r$ with edges connecting nearest neighbours in $\Lambda_r$.
Periodic boundary conditions (PBC) correspond to $\Gcal=\T_r$
with edges connecting nearest neighbours in $\T_r$, which include all edges of $\Lambdar$ and additional edges that join opposite sides of $\Lambda_r$.
The \emph{two-point functions} are defined by
\begin{equation}
    \tau^{\Lambda_r}_\beta(x,y) = \langle \sigma_x \sigma_y \rangle_\beta^{\Lambda_r},
    \qquad
    \tau^{\T_r}_\beta(x,y) = \langle \sigma_x \sigma_y \rangle_\beta^{\T_r} ,
\end{equation}
and the finite-volume \emph{susceptibilities} are defined by
\begin{equation}
    \chi^{\Lambda_r}(\beta) = \sum_{x\in \Lambda_r} \tau^{\Lambda_r}_\beta(0,x),
    \qquad
    \chi^{\T_r}(\beta) = \sum_{x\in \T_r} \tau^{\T_r}_\beta(0,x).
\end{equation}

It follows from the Griffiths inequalities
that $\taub^\Lambdar$ is nonnegative and increasing in $r$, so that the infinite-volume two-point function $\taub(x,y) = \lim_{r\to\infty} \taub^\Lambdar(x,y)$ and infinite-volume
susceptibility $\chi(\beta) = \lim_{r\to\infty} \chi^\Lambdar(\beta)$ exist in $[0, \infty]$.
Moreover, it is known that there is a critical inverse temperature $\beta_c$,
below which $\taub(x,y)$ decays exponentially as $\abs{x-y}\to \infty$ and $\chib$ is finite.
The susceptibility is infinite at $\beta_c$, and
from \cite{Aize82} we know that for $d>4$ and $\beta < \beta_c$,
\begin{equation} \label{eq:chibds}
\chi(\beta) \asymp \frac{1}{\beta_c-\beta}.
\end{equation}
An essential ingredient in our analysis is the very recent result from
\cite{DP25-Ising} that, for $d > 4$, there are positive constants $c_0,C_0$
such that for all $\beta \le \beta_c$,
\begin{alignat}2
\label{eq:DPub}
\taub(0,x) &\le  \frac{C_0}{\xvee^{d-2}} e^{-c_0(\beta_c-\beta)^{1/2}\norm x_\infty}
	&&\qquad\qquad(x \in \Z^d) ,	\\
\label{eq:DPlb}
\taub(0,x) & \ge \frac{c_0}{ \xvee^{d-2}}
	 &&(\norm x_\infty \le c_0(\beta_c-\beta)^{-1/2}).
\end{alignat}
For the spread-out Ising model in dimensions $d>4$, or for the nearest-neighbour model in sufficiently
high dimensions, more precise asymptotics for the critical two-point function have
been proved using the lace expansion \cite{Saka07}.

The influence of boundary conditions on finite-size critical behaviour in
dimensions $d \ge 4$ has been studied extensively in the physics literature
for Ising and related models,
e.g., \cite{BEHK22,DGGZ22,DGGZ24,LB97,LM16,WY14,ZGFDG18}.
In particular, it has been observed numerically and via physics scaling
arguments that, for $d>4$,
\begin{align}
\label{eq:tauchi}
    \tau^{\Lambda_r}_{\beta_c}(0,x) \asymp \frac{1}{\xvee^{d-2}},
    \quad
    \tau^{\T_r}_{\beta_c}(0,x) \asymp \frac{1}{\xvee^{d-2}} + \frac{1}{r^{d/2}},
    \qquad
    \chi^{\Lambda_r}(\beta_c) \asymp r^2,
    \quad
    \chi^{\T_r}(\beta_c) \asymp r^{d/2}.
\end{align}
(The claim for $\tau^{\Lambda_r}_{\beta_c}(0,x)$ is for $x$ not too close to the boundary of the box.)
The constant term $r^{-d/2}$ in
$\tau^{\T_r}_{\beta_c}(0,x)$
is referred to as the \emph{plateau}, and it is responsible for the larger susceptibility for PBC compared to FBC.

In recent years, the plateau for PBC has been proven to exist for
simple random walk in dimensions $d>2$ \cite{Slad23_wsaw,DGGZ24},
self-avoiding walk for $d>4$  \cite{Slad23_wsaw,Liu24},
(spread-out) percolation for $d>6$ \cite{HS14,HMS23}, and
(spread-out) lattice trees and lattice animals for $d>8$ \cite{LS25a}.
A theory of the effect of FBC vs PBC
in the setting of the weakly-coupled hierarchical $n$-component
$|\varphi|^4$ model in dimensions $d \ge 4$ was developed in \cite{MPS23,PS25},
and we comment further on this below.
A general theory is outlined in \cite{LPS25-universal}.

For the Ising model in dimensions $d>4$, the absence of a plateau for FBC
is an immediate consequence of \eqref{eq:DPub} and the monotonicity of
$\tau^{\Lambda_r}_{\beta_c}(0,x)$ in $r$.  A matching lower bound
for FBC is proved in \cite{CJN21}.
For PBC,
incomplete steps were
taken in \cite{Papa06} towards proving existence of the plateau.
 In this paper, we build on the method of
\cite{Papa06} to prove \eqref{eq:tauchi} for PBC in dimensions $d>4$.
We also prove the non-vanishing of the renormalised coupling constant under PBC, indicating
a non-Gaussian limit for the average spin.  This can be contrasted with the Gaussian limit under FBC obtained
in \cite{Aize82,Froh82,AD21,Pani24}.

We emphasise that we do not use the lace expansion in any way, and our
results do not require a small parameter.  The input from the lace expansion
used to establish plateau theorems
for self-avoiding walk, percolation, and lattice trees and lattice animals
in \cite{Slad23_wsaw,Liu24,HS14,HMS23,LS25a}
is replaced here by the
near-critical bounds \eqref{eq:DPub}--\eqref{eq:DPlb}.

\subsection{Main results}

Our first theorem indicates that for $\beta$ in an interval slightly below $\beta_c$, the torus and infinite-volume two-point functions are related by a plateau term $\chi(\beta)/r^d$.

\begin{theorem}
\label{thm:plateau}
Let $d>4$.
There are constants $c_i>0$, depending on $d$ but not on $\beta$ or $r$,
such that
\begin{equation}
\label{eq:plateau-ub}
    \tau_{\beta}^{\T_r}(0,x)
    \le
    \tau_\beta(0,x) +c_1\frac{\chi(\beta)}{r^d}
\end{equation}
for all $r \ge 3$, all $\beta \in (0,\beta_c)$,
and all $x \in \Tr$, and
\begin{equation}
\label{eq:plateau-lb}
    \tau_{\beta}^{\T_r}(0,x)
    \ge
    \tau_\beta(0,x) +c_2 \frac{\chi(\beta)}{r^d}
\end{equation}
for all $r\ge 3$, all
$\beta \in [ \beta_c- c_3r^{-2},   \beta_c-c_4 r^{-d/2}]$ and all $\|x\|_\infty \ge c_5$.
Moreover, the constant $c_3$ is given by $c_3 = \frac 49 c_0^2$ with $c_0$ the
constant of \eqref{eq:DPlb}.
\end{theorem}

Note that for the lower bound, the interval of $\beta$ values is nonempty when $r$ is sufficiently large, since $d > 4$.
In this case, we write
\begin{equation}
    \beta_*  = \beta_*(r)  = \beta_c-c_4 r^{-d/2}.
\end{equation}
At $\beta_*$, we have $\chi(\beta_*) \asymp r^{d/2}$ by \eqref{eq:chibds},
so \eqref{eq:plateau-ub}--\eqref{eq:plateau-lb} identify a plateau term in $\tau_{\beta_*}^\Tr$ of order $r^{-d/2}$.
The monotonicity in $\beta$ then implies that the plateau term
persists for larger $\beta$, in particular for $\beta_c$.
\begin{Coro} \label{coro:plateau}
Let $d>4$ and $r$ be large. Then
\begin{align} \label{eq:beta*plateau}
\tau_{\beta_*}^{\T_r}(0,x) \asymp \nnnorm x^{-(d-2)} + r^{-d/2}
\end{align}
for all $x \in \T_r$.
In particular, $\tau_\beta^{\T_r}(0,x) \gtrsim \nnnorm x^{-(d-2)} + r^{-d/2}$
for all $\beta \ge \beta_*$.
\end{Coro}

Summation of \eqref{eq:beta*plateau} over $x \in \T_r$
immediately gives the following bounds on the susceptibility,
because $\sum_{x\in \Tr} \nnnorm x^{-(d-2)} = O(r^2) = o(r^{d/2})$
in dimensions $d>4$.

\begin{Coro} \label{coro:chi}
For $d>4$ and $r$ large,
\begin{equation}
    \chi^{\mathbb T_r}(\beta_*) \asymp r^{d/2}.
\end{equation}
In particular, $\chi^{\mathbb T_r}(\beta) \gtrsim r^{d/2}$ for all $\beta \ge \beta_*$.
\end{Coro}

Let $S_r = r^{-d}\sum_{x\in \T_r}\sigma_x$ denote the average spin on the torus.
We define the \emph{renormalised coupling constant} by
\begin{equation}
\label{eq:g-def}
    g^{\T_r}(\beta) =
    -
    \frac{\langle S_r^4 \rangle_\beta^{\T_r} - 3 (\langle S_r^2 \rangle_\beta^{\T_r})^2}
    {(\langle S_r^2 \rangle_\beta^{\T_r})^2}
    .
\end{equation}
By the Lebowitz inequality, the above numerator is non-positive, so $g^{\T_r}(\beta) \ge 0$.
By the second Griffiths inequality, $\langle S_r^4 \rangle_\beta^{\T_r} \ge
(\langle S_r^2 \rangle_\beta^{\T_r})^2$, so $g^{\T_r}(\beta) \le 2$.
The following theorem shows that $g^{\T_r}(\beta_*)$ is bounded away from zero.
This indicates a non-Gaussian
limit for the average field at $\beta=\beta_*$: for a Gaussian random variable
the numerator in \eqref{eq:g-def} is zero.
This is
in contrast to the situation at (and below) $\beta_c$ with free boundary conditions \cite{Aize82,AD21,Froh82,Pani24}, where the limit is Gaussian.
We do not have a proof that $g^{\mathbb T_r}$ is increasing
in $\beta$, but see \eqref{eq:glim} below.

\begin{theorem}
\label{thm:gr}
Let $d>4$.
There is a constant $c_g >0$ such that for all $r$ large,
\begin{equation}
    0 < c_g \le g^{\T_r}(\beta_*) \le 2 .
\end{equation}
\end{theorem}

Our lower bounds on
the susceptibility in Corollary~\ref{coro:chi},
on the two-point function in Corollary~\ref{coro:plateau},
and on the renormalised coupling constant in Theorem~\ref{thm:gr},
are claimed in \cite{Papa06},
but the arguments put forth in \cite{Papa06} do not constitute a proof.
Nevertheless, a useful strategy is proposed in \cite{Papa06}, and we build
upon that strategy in our work.

\subsection{Conjectured asymptotic formulas}
\label{sec:conjectures}

A precise and detailed account of the effect of boundary conditions on
finite-size scaling is given in \cite{MPS23,PS25}, via a rigorous renormalisation
group analysis.
The results of  \cite{MPS23,PS25} are proved for the weakly-coupled
$n$-component $|\varphi|^4$ model on the $d$-dimensional hierarchical lattice
for $d \ge 4$ and for all $n \ge 1$,
but, on the basis of universality, the results obtained there for the $1$-component
model are conjectured to apply also to
all $1$-component spin models in the universality
class of the nearest-neighbour Ising model.  An antecedent of the
conjectures goes back to \cite{BZ85} (see also \cite[Chapter~32]{Zinn21}).
The conjectures present predictions for
a specific and universal non-Gaussian limit for the average field,
and precise asymptotic amplitudes for the susceptibility, the plateau,
and the renormalised coupling constant, throughout a critical
window of width $r^{-d/2}$ (with a logarithmic correction for $d=4$) centred
at $\beta_c$.
We define the \emph{window scale} and the \emph{large-field scale}, respectively, by
\begin{equation}
\label{eq:window_choice}
	w_r =
    \begin{cases}
        a_{4} (\log r)^{-1/6}  r^{-2} & (d=4)
        \\
        a_{d}\,   r^{-d/2}
        & (d>4),
    \end{cases}
    \qquad
	h_r =
    \begin{cases}
        b_{4} (\log r)^{1/4}  r^{-1} & (d=4)
        \\
        b_{d}\,   r^{-d/4}
        & (d>4),
    \end{cases}
\end{equation}
with suitably chosen constants $a_d,b_d>0$.
We also define the \emph{universal profile} $f:\R \to (0,\infty)$ by
\begin{equation}
    f(s)
    =
    \int_{\R} x^2 \mathrm d\sigma_s,
    \qquad
    \mathrm d\sigma_s
    =
    \frac{ e^{-\frac 14 x^4 + \frac s2 x^2} \mathrm dx}
    {\int_{\R} e^{-\frac 14 x^4 + \frac s2 x^2} \mathrm dx}
    .
\end{equation}
The function $f$ is strictly monotone increasing, and satisfies
$f(s) \sim |s|^{-1}$ as $s \to -\infty$ and $f(s) \sim s$ as $s \to \infty$
(see \cite[Section~1.5.2]{MPS23}
where $f_1(s)=f(-s)$).
The following conjecture proposes a precise
non-Gaussian limit for the average field $S_r$, and a precise
role for the universal profile $f$ in the
critical finite-size scaling of the susceptibility, the two-point function, and the
renormalised coupling constant for the Ising model.
The conjecture is a special case ($n=1$) of conjectures from
\cite{MPS23,PS25}
for $n$-component models for all $n \ge 0$ ($n=0$ corresponds to the self-avoiding walk).  The profile $f$ also arises in the Curie--Weiss model; a recent example is
\cite[Theorem~1.1]{BBE24}.

\begin{Conj}[Non-gaussian limit and universal profile]
\label{conj:chi}
For $d \ge 4$, $s \in \R$, and $\beta = \beta_c + s w_r$,
as $r \to \infty$ the rescaled average field
converges in distribution to a
real-valued random variable
with the non-Gaussian distribution $\mathrm d\sigma_s$:
\begin{equation}
    h_r^{-1} S_r
    \Rightarrow \mathrm d\sigma_s
    \quad
    \text{under $\langle \cdot\rangle_{\beta_c+sw_r}^{\Tr}$.}
\end{equation}
Moments of $h_r\inv S_r$ also converge to moments of $\mathrm d\sigma_s$. In particular,
\begin{gather}
    \chi^{\T_r}(\beta_c + sw_r)
    = r^d \langle S_r^2\rangle_{\beta_c+sw_r}^{\mathbb T_r}
    \sim
    r^d h_r^2 f(s)
    =
    b_d^{2} \, f(s)
    \times
    \begin{cases}
        (\log r)^{1/2} r^2 & (d=4)
        \\
        r^{d/2} & (d>4),
    \end{cases}
\\
\label{eq:glim}
    \lim_{r\to\infty} g^{\T_r}(\beta_c+sw_r)
    =
    3 -
    \frac{\int_{\R} x^4 \mathrm d\sigma_s}{(\int_{\R} x^2 \mathrm d\sigma_s)^2}.
\end{gather}
For $x \in \Tr$, the plateau is asymptotically the susceptibility divided by
the volume:
\begin{align}
\label{eq:conj-plateau}
    \tau^{\T_r}_{\beta_c + sw_r}(0,x)
    -
    \tau_{\beta_c}(0,x)
    \sim h_r^2 f(s)
    =
    b_d^{2} \,
     f(s) \times
    \begin{cases}
        (\log r)^{1/2} r^{-2} & (d=4)
        \\
        r^{-d/2} & (d>4).
    \end{cases}
\end{align}
\end{Conj}

As observed in \cite[Section~1.5.3]{MPS23},
the right-hand side of \eqref{eq:glim} is strictly monotone increasing,
and has limits $0$ and $2$ as $s \to -\infty$ and $s \to +\infty$ respectively.
Thus, assuming \eqref{eq:glim}, the
renormalised coupling constant varies monotonically over its possible range $[0,2]$ as the
window is traversed.
At the infinite-volume critical point, \eqref{eq:glim} gives
\begin{equation}
    \lim_{r\to\infty} g^{\T_r}(\beta_c)
    =
    3 - 4\left( \frac{\Gamma(5/4)}{\Gamma(3/4)} \right)^2
    =
    0.81156\ldots
    \qquad
    (d \ge 4).
\end{equation}

The results of \cite{MPS23,PS25} further suggest the following conjecture
for the Ising model with FBC in dimensions $d \ge 4$.

\begin{Conj}[Free boundary conditions]
\label{conj:FBC}
There is an \emph{effective critical point} $\beta_{c,r}$ for the Ising
model with free boundary conditions on the box $\Lambda_r$, which obeys
\begin{equation}
    \beta_{c,r} = \beta_c + v_r, \qquad
    v_r \asymp
    \begin{cases}
        (\log r)^{1/3}r^{-2} & (d=4)
        \\
        r^{-2} & (d>4),
    \end{cases}
\end{equation}
such that the statements of Conjecture~\ref{conj:chi}
hold
as stated when $\beta_c$ is replaced by $\beta_{c,r}$.
In particular, we conjecture that
the plateau phenomenon also occurs for free boundary conditions
at the effective critical point.
\end{Conj}

The shift scale $v_r$ exceeds the window scale $w_r$, so $\beta_c$ is \emph{not} inside the conjectured scaling window for FBC.
This is consistent with theorems proving a Gaussian limit at $\beta_c$ under FBC \cite{Aize82,Froh82,AD21,Pani24,MPS23,PS25}.

Closely related conjectures for percolation are discussed in Appendix~\ref{sec:perc}.
Our analysis is based on the random current representation for the Ising model,
which gives a geometric representation for correlation functions with
similarities to bond percolation.  In view of this parallel with percolation,
the following intuitive picture of the percolation plateau
provides intuition more generally.

\begin{paragraph}{Origin of the plateau.}
For spread-out independent Bernoulli bond percolation on the torus $\T_r$
in dimensions $d >6$,
it is proved in \cite{HMS23} that, in a window of width $r^{-d/3}$ containing
the infinite-volume critical point $p_c$, the two-point function behaves
as $|x|^{-(d-2)}+r^{-2d/3}$ and the susceptibility (expected cluster size)
behaves as $r^{d/3}$.
Also,
it has been proved for sufficiently high dimensions (presumably true for all $d>6$)
that at $p_c$
the largest cluster on $\T_r$ is of order $r^{2d/3}$, whereas with FBC
the largest cluster has size $r^4$; see \cite[Theorem~13.5]{HH17book}, \cite[Theorem~13.22]{HH17book}, \cite{Aize97}.
For PBC, if we assume that two points $0$ and $x$ in $\T_r$ are connected to
a largest cluster (in particular connected to each other)
with probability $(r^{2d/3}/r^d)^2$, we find a plateau term of
the correct order $r^{-2d/3}$.  For FBC, a similar computation gives a probability
$(r^4/r^d)^2 = r^{-(d-2)+(d-6)}$, which is negligible compared to
$\tau_{p_c}(0,x) \asymp |x|^{-(d-2)}$.
This rough computation is consistent with a plateau for PBC and no plateau for FBC.

The equality of powers $r^{2d/3}$ and $r^4$ when $d=6$ indicates a special role
for dimension $6$.  No plateau is predicted for $d<6$, or indeed for other
models such as spin systems below their upper critical dimension.
For percolation in dimensions below $6$, hyperscaling
plays a role \cite{Grim99}.  Hyperscaling reflects the dimension-dependent
manner in which critical clusters are constrained to fit into space.
The distinction between a plateau above the critical dimension, and no plateau
below the critical dimension, may be understood intuitively as due to the fact
that in high dimensions a critical cluster that ``wraps around'' discovers new
vertices, whereas in low dimensions the wrapping mostly encounters vertices already
present in the cluster without wrapping.  In \cite{DGGZ22}, the role of wrapping
above the upper critical dimension is studied numerically in related settings.
\end{paragraph}

\subsection{Structure of proof}

We now state five propositions and prove our main results subject to these propositions. The proof of Theorem~\ref{thm:plateau} mirrors the general strategy used for self-avoiding walk in \cite{Slad23_wsaw, Liu24} and for percolation in \cite{HMS23}.
However, there are model-specific elements in the proof of these propositions.
For the Ising model, we will use
\eqref{eq:DPub}--\eqref{eq:DPlb} as well as the random current representation
\cite{Aize82,GHS70}.

We observe that $\taub$ and $\taub^\Tr$ are translation invariant (but $\taub^\Lambdar$ is not), so we write $\taub(x) = \taub(0,x)$ and $\taub^\Tr(x)= \taub^\Tr(0,x)$ from now on.
For $x\in \Zd$, we define
\begin{equation}
\label{eq:Tsum}
    \Tsum_\beta^{(r)}(x) = \sum_{u \in \Z^d} \taub(x+ru) .
\end{equation}

\begin{proposition} \label{prop:unfold}
Let $d\ge1$, $r\ge 3$, $\beta>0$, and $x\in \Tr$. Then
$\taub^\Tr(x)  \le  \Tsum_\beta^{(r)}(x) .$
\end{proposition}

Proposition~\ref{prop:unfold} suggests decomposing $\taub^\Tr(x)$ as
\begin{align} \label{eq:taupsi1}
\taub^\Tr(x) =
\Tsum_\beta^{(r)}(x)  - \big( \Tsum_\beta^{(r)}(x) - \taub^\Tr(x) \big).
\end{align}
We estimate the two terms on the right-hand side separately as follows.
Recall that $\chi(\beta)$ denotes the $\Zd$ susceptibility.

\begin{proposition} \label{prop:psiub}
Let $d>4$.
There is a constant $c_1>0$ such that
for all $r\ge 3$ and $\beta < \beta_c$,
\begin{align}
\Tsum_\beta^{(r)}(x)
\le \taub(x) + c_1 \frac \chib {r^d}
\qquad (x\in \Tr).
\end{align}
\end{proposition}

\begin{proposition} \label{prop:psilb}
Let $d>4$ and $c_3=\frac 49 c_0^2$,
where $c_0$ is the constant of \eqref{eq:DPlb}.
There is a constant $c_T > 0$
such that
for all $r\ge 3$ and all $\beta\in [\beta_c-c_3r^{-2},\beta_c)$,
\begin{align}
\Tsum_\beta^{(r)}(x)  \ge  \taub(x) + c_T \frac \chib {r^d}
\qquad (x\in \Tr).
\end{align}
\end{proposition}

\begin{proposition} \label{prop:psipsitil}
Let $d>4$.  Let $c_T$ be the constant of Proposition~\textup{\ref{prop:psilb}}.
There are constants $c_4, c_5 >0$ such that
for all $r\ge 3$ and $\beta \le \beta_c-c_4r^{-d/2}$,
\begin{align}
\Tsum_\beta^{(r)}(x) -  \tau_\beta^{\T_r}(x)
\le \half c_T \frac{\chi(\beta)}{r^d}
\qquad (x\in \Tr,\;\norm x_\infty \ge c_5).
\end{align}
\end{proposition}

\begin{proof}[Proof of Theorem~\textup{\ref{thm:plateau}}]
The upper bound follows immediately from Propositions~\ref{prop:unfold} and \ref{prop:psiub}.
Recall that $\beta_*=\beta_c-c_4r^{-d/2}$.
For the lower bound, let $\beta\in [\beta_c-c_3r^{-2},\beta_*]$
and $\norm x_\infty \ge c_5$.
Then by \eqref{eq:taupsi1}, and by Propositions~\ref{prop:psilb} and \ref{prop:psipsitil},
\begin{align}
\taub^\Tr(x)  \ge  \taub(x) + \( 1 - \half \) c_T \frac \chib {r^d} .
\end{align}
This gives the desired result with $c_2 = \half c_T$.
\end{proof}

\begin{proof}[Proof of Corollary~\textup{\ref{coro:plateau}}]
The upper bound follows from \eqref{eq:plateau-ub}, $\tau_\beta(x) \lesssim \nnnorm x^{-(d-2)}$ by \eqref{eq:DPub}, and the fact that $\chi(\beta_*) \asymp r^{d/2}$ from \eqref{eq:chibds}.

For the lower bound, when $\norm x_\infty \ge c_5$,
we first observe that the lower bound \eqref{eq:DPlb} applies
at $\beta=\beta_*$ provided that $r/2\leq c_0(\beta_c-\beta^*)^{-1/2}$, i.e. $r^{(d-4)/4} \ge (2c_0)^{-1} c_4^{1/2}$.
It then follows from
\eqref{eq:plateau-lb} and $\chi(\beta) \ge \chi(\beta_*) \gtrsim r^{d/2}$
that
\begin{equation}
    \tau_{\beta}^{\T_r}(x)
    \ge
    \tau_{\beta_*}(x) + c_2 \frac{ \chi(\beta_*) }{r^d}
    \gtrsim \nnnorm x^{-(d-2)} + r^{-d/2},
\end{equation}
as desired.
When $\norm x_\infty \le c_5$, we let $m_5=\min_{\norm x_\infty\le c_5}\tau_{\beta_c/2}^{\Lambda_{2c_5}}(0,x) >0$.
By the second Griffiths inequality, when $r$ is large we have
\begin{equation}
    \tau_{\beta_*}^{\T_r}(x)
    \ge
    \tau^{\Lambda_{2c_5}}_{\beta_*}(0,x)
    \ge
    m_5
    \ge
    \frac{m_5}{2} \frac{1}{\nnnorm x^{d-2}} +   \frac{m_5}{2}\frac{1}{r^{d/2}},
\end{equation}
which completes the proof.
\end{proof}

The proof of Theorem~\ref{thm:gr} uses the following proposition,
which we prove in Appendix~\ref{sec:u4proof}.
The torus Ursell function $U_{4,\beta}^\Tr$ is defined by
\begin{equation}
    U_{4,\beta}^\Tr(0,x,y,z)
    =
    \langle \sigma_0\sigma_x\sigma_y\sigma_z \rangle_\beta^\Tr
    -
    \langle \sigma_0\sigma_x\rangle_\beta^\Tr
    \langle \sigma_y\sigma_z \rangle_\beta^\Tr
    -
    \langle \sigma_0\sigma_y \rangle_\beta^\Tr
    \langle \sigma_x\sigma_z \rangle_\beta^\Tr
    -
    \langle \sigma_0\sigma_z \rangle_\beta^\Tr
    \langle \sigma_x\sigma_y \rangle_\beta^\Tr  .
\end{equation}
By the Lebowitz and Griffiths inequalities, $U_{4,\beta}^\Tr(0,x,y,z) \in [-2, 0]$.
The torus \emph{bubble diagram} is defined by
\begin{equation}
\bubble(\beta) = \sum_{x\in\T_r} \taub^\Tr(x)^2.
\end{equation}

\begin{proposition} \label{prop:u4}
Let $d\ge 1$.
For all $r\geq 3$
and for all $\beta>0$,
\begin{equation} \label{eq:u4claim}
\sum_{x,y,z\in\T_r}  |U_{4,\beta}^\Tr (0,x,y,z)|
\ge \frac{1}{128d^2} \cdot \min \Bigl\{
	\beta(\chi^\Tr)^2  \frac{\del \chi^\Tr}{\del \beta} ,
	\frac 1 \bubble \Big( \frac{\del \chi^\Tr}{\del \beta} \Big)^2
	\Big\} .
\end{equation}
\end{proposition}

\begin{proof}[Proof of Theorem~\textup{\ref{thm:gr}}]
All quantities in this proof are torus quantities, so
to simplify the notation we drop the label $\T_r$ everywhere.
By a direct computation using $S_r = r^{-d}\sum_{x\in\T_r}\sigma_x$,
\begin{equation}
    g(\beta) =
    -
    \frac{\langle S_r^4 \rangle_\beta - 3 (\langle S_r^2 \rangle_\beta)^2}
    {(\langle S_r^2 \rangle_\beta)^2}
    =
    \frac{1}{\chi(\beta)^2 r^{d}} \sum_{x,y,z}|U_{4,\beta} (0,x,y,z)|
    .
\end{equation}
We use Proposition~\ref{prop:u4} to bound from below the sum on the right-hand side.
A lower bound on the derivative of the susceptibility,
essentially due to \cite{AG83} (see \cite[(12.50)]{FFS92}),
states that
\begin{equation}
    \beta \frac{\partial\chi}{\partial\beta} \ge
    \frac{2d \chi^2}{1+(2d)^2\bubble}
    \left[
    1 - \frac{4d\bubble}{\chi} - \frac{\bubble}{\chi^2} \right]
    .
\end{equation}
It follows from the upper bound of Corollary~\ref{coro:plateau} that $\bubble(\beta_*) \lesssim 1$.
Combined with $\chi(\beta_*) \asymp r^{d/2}$ from Corollary~\ref{coro:chi}, we get
\begin{equation}
g(\beta_*)
\gtrsim \frac 1 {r^{d}} \chi(\beta_*)^2
\asymp 1 ,
\end{equation}
and the proof is complete.
\end{proof}

\section{Proof of Propositions~\ref{prop:unfold}--\ref{prop:psipsitil}}

We now prove Propositions~\ref{prop:unfold}--\ref{prop:psipsitil}.
The proofs of Propositions~\ref{prop:unfold} and \ref{prop:psipsitil}
depend on new diagrammatic estimates that we state as Theorem~\ref{thm:coupling_bounds}.
We begin with the proofs of Propositions~\ref{prop:psiub} and \ref{prop:psilb}, which are easier.

\subsection{Proof of Propositions~\ref{prop:psiub} and \ref{prop:psilb}}

These two propositions provide upper and lower bounds on
\begin{equation}
    \Tsum_\beta^{(r)}(x) = \sum_{u \in \Z^d} \taub(x+ru) ,
\end{equation}
respectively.
It is useful to observe that,
since $\norm x_\infty \le \frac r 2$ for all $x\in \Tr$,
for all $0 \ne u \in \Z^d$ we have
\begin{align} \label{eq:xuub}
\norm{ x+ru }_\infty
\le \norm{ ru }_\infty + \frac r 2
\le \frac 3 2 r \norm u_\infty ,
\\
\label{eq:xulb}
\norm{ x+ru }_\infty
\ge \norm{ ru }_\infty - \frac r 2
\ge \frac 1 2 r \norm u_\infty .
\end{align}
Proposition~\ref{prop:psiub} follows simply from
\eqref{eq:xulb} and
the upper bound \eqref{eq:DPub} on
$\taub$, and we defer its proof to Appendix~\ref{sec:convolution}
where it is included in the elementary Lemma~\ref{lem:GT3}.
We prove Proposition~\ref{prop:psilb} here, using
\eqref{eq:xuub} and the lower bound \eqref{eq:DPlb} on $\taub$.

\begin{proof}[Proof of Proposition~\textup{\ref{prop:psilb}}]
For the proof, we wish to apply the lower bound \eqref{eq:DPlb} to all $x +ru$ with
$x\in \Lambda_r$ and $\|u\|_\infty \le M$, with
a well-chosen
$M \ge 1$.
Since \eqref{eq:DPlb} holds when $\|x+ru\|_\infty \le c_0(\beta_c-\beta)^{-1/2}$,
it suffices to have
\begin{equation}\label{eq:proof1.11}
    M \ge  1, \qquad \frac r2 + rM \le c_0(\beta_c-\beta)^{-1/2}.
\end{equation}
We try $M=M(\beta,r) = \mu c_0 (\beta_c - \beta)^{-1/2} r\inv$.
The two inequalities of \eqref{eq:proof1.11} then become
\begin{equation}
    r^2(\beta_c-\beta) \le   \mu^2 c_0^2,
    \qquad r^2(\beta_c-\beta) \le 4(1-\mu)^2 c_0^2 .
\end{equation}
The choice $\mu=2/3$ makes the two right-hand sides equal, and
in that case both inequalities assert $r^2(\beta_c - \beta) \le \frac 4 9 c_0^2$.
Thus, if we choose $c_3=\frac 49 c_0^2$ and require $\beta \in [\beta_c-c_3r^{-2},\beta_c)$, then
by using \eqref{eq:DPlb} and \eqref{eq:xuub} in the definition of $\Tbr(x)$
we obtain
\begin{align}
\Tbr(x) - \taub(x)
\gtrsim \sum_{1 \le \norm u_\infty \le M} \frac 1 {\norm {x+ru}_\infty^{d-2}}
\ge \sum_{1 \le \norm u_\infty \le M} \frac 1 { ( \frac 3 2 r \norm {u}_\infty )^{d-2}}
\asymp \frac {M^2} {r^{d-2}}.
\end{align}
By the definition of $M$ and the fact that
$\chi(\beta) \asymp (\beta_c - \beta)\inv$ (by \eqref{eq:chibds}), the right-hand
side is bounded from below by a multiple of $r^{-d}\chi(\beta)$.  This completes the proof.
\end{proof}

\subsection{Diagrammatic estimates}

The proofs of Propositions~\ref{prop:unfold} and \ref{prop:psipsitil} rely on
the following theorem, whose proof is our main work.
The proof of Theorem~\ref{thm:coupling_bounds}
is based on the random current representation and a new coupling of the Ising model on the
torus and on $\Z^d$, and is deferred to Section~\ref{sec:currents}.
For the statement of Theorem~\ref{thm:coupling_bounds},
given $x,x'\in \Z^d$ we write $x \cong x'$ if $x-x'\in r\Z^d$, and
we write $\pi x$ for the unique point in $\Lambdar$ (which we identify with $\Tr$) such that $\pi x \cong x$.
The bound \eqref{eq:upper_bound} repeats the statement of
Proposition~\ref{prop:unfold}.

\begin{theorem} \label{thm:coupling_bounds}
Let $d\ge 1$, $r\ge 3$, and $\beta >0$.
For any $x\in \Tr$,
\begin{equation} \label{eq:upper_bound}
\tau_\beta^{\T_r}(x) \le \Tsum_\beta^{(r)}(x),
\end{equation}
\begin{equation} \label{eq:lower_bound_diagram}
\Tsum_\beta^{(r)}(x) - \taub^\Tr(x)
\le \sum_{ x' \cong x } \sum_{y} \sum_{ \substack{y' \cong y \\ y' \ne y} } \sum_{z } \taub(z) \taub(y-z) \taub(y'-z) \taub^\Tr(\pi y' -\pi z) \taub(x'-y) ,
\end{equation}
where $x',y,y',z$ range over $\Zd$ with the indicated restrictions.
\end{theorem}

The following corollary, which gives an upper bound that simplifies the presentation of \eqref{eq:lower_bound_diagram}, will be important.
We define the torus $\star$ convolution
for functions $f,g : \T_r \to \R$ (periodic functions on $\Z^d$) by
\begin{equation}
\label{eq:starconv}
    (f\star g)(x) = \sum_{y\in \T_r} f(y)g(x-y)
    \qquad (x \in \T_r).
\end{equation}
We also use the $\Zd$ convolution $(f*g)(x)=\sum_{y\in \Z^d} f(y)g(x-y)$ for
absolutely summable
functions on $\Z^d$, so it is important to mark the difference between $\star$
and $*$.

\begin{Coro} \label{cor:lower_bound}
Let $d\ge 1$, $r\ge 3$, and $\beta >0$.
For any $x\in \Tr$,
\begin{equation} \label{eq:lower_bound_T}
\Tsum_\beta^{(r)}(x) - \taub^\Tr(x)
\le
2\Big(\Tbr \star \big[\big(\Tbr\big)^2 [\Tbr -(\taub \circ \pi)]\big] \star \Tbr \Big) ( x) .
\end{equation}
\end{Coro}

\begin{proof}
We use \eqref{eq:upper_bound} and the periodicity of $\Tbr$ to get
$\taub^\Tr(\pi y' -\pi z) \le \Tbr (y - z)$.
We also rewrite the sums over $x',y'$ using the definition of $\Tbr$.
These allow us to bound the right-hand side of \eqref{eq:lower_bound_diagram} by
\begin{equation} \label{eq:lower_bound_pf10}
\sum_{y} \sum_{z } \taub(z) \taub(y-z) \big( \Tbr - \taub \big)(y-z) \Tbr(y - z) \Tbr(x-y) .
\end{equation}
Since every point in $\Zd$ can be written uniquely
as a point in $\Lambda_r + r\Zd$, we can write
\begin{equation}
x = \check x, \quad
y = \check y + ru, \quad
z = \check z + rv \qquad
(\check x, \check y, \check z \in \Lambda_r, \
u,v\in \Zd),
\end{equation}
so the two sums over $y,z$ become four sums over $\check y, \check z, u,v$.
Since $\Tbr$ is periodic, we have $\Tbr(y - z) = \Tbr( \check y - \check z)$ and $\Tbr(x-y) = \Tbr(\check x - \check y)$, independent of $u$.
Hence, the sum over $u$ gives the upper bound
\begin{equation}
\sum_u \big( \taub\Tbr - \taub^2 \big) \big( \check y - \check z + r(u-v) \big)
\le \Tbr( \pi( \check y - \check z) )^2 - \taub( \pi( \check y- \check z) )^2 ,
\end{equation}
where in the subtracted term we kept only one term in the sum as an upper bound.
We further bound the above difference of squares by
$2\Tbr [\Tbr  - (\taub\circ \pi)]$.
Finally, we complete the proof with the observation
that  $\sum_v \taub(\check z + rv) = \Tbr(\check z)$.
\end{proof}

\begin{remark}
It is interesting to compare the diagrammatic upper bound
\eqref{eq:lower_bound_T} with diagrams in related models.
The counterpart of $\Tbr - \taub^\Tr$ is denoted by $\psi - \psi^\T$ in \cite{Slad23_wsaw, HMS23, Liu24}.
 For weakly or strictly self-avoiding
walk, the bounds \cite[(6.36)]{Slad23_wsaw} and \cite[(4.20)]{Liu24}  can be written as
\begin{equation}\label{eq:diagram saw}
T \star [ \delta \cdot (T - \tau \circ \pi) ] \star T ,
\end{equation}
with $\delta$ the Kronecker delta.
For percolation, the bounds of \cite[(4.35), (4.42)]{HMS23} can be combined into
the upper bound
\begin{equation}\label{eq:diagram perco}
T \star [ T \cdot (T - \tau \circ \pi) ] \star T .
\end{equation}
For Ising, we found the bound
\begin{equation}\label{eq:diagram ising}
T \star [ T^2 \cdot (T - \tau \circ \pi) ] \star T .
\end{equation}
The $T^2$
in the Ising bound, compared to the $T$ in percolation,
is responsible for lowering the critical dimension from $6$ to $4$.
The three upper bounds are depicted in Figure~\ref{figure:diagrams other models}.
\end{remark}

\begin{figure}\begin{center}
\includegraphics{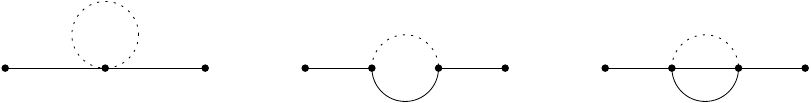}
\caption{From left to right, diagrammatic representations of the upper bounds
for self-avoiding walk in \eqref{eq:diagram saw}, for percolation in
\eqref{eq:diagram perco}, and for the Ising model in \eqref{eq:diagram ising}.
The dotted line corresponds to $T-\tau\circ \pi$, which produces a factor
$\chi/r^d$ for each model.}
\label{figure:diagrams other models}	
\end{center}
\end{figure}

\begin{proof}[Proof of Proposition~\textup{\ref{prop:psipsitil}}]
It follows from Corollary~\ref{cor:lower_bound} and Proposition~\ref{prop:psiub} that
\begin{equation} \label{eq:lower_bound_T2}
\Tsum_\beta^{(r)}(x) - \taub^\Tr(x)
\lesssim
\frac{\chi(\beta)}{r^d}
\Big(\Tbr \star \big(\Tbr\big)^2  \star \Tbr \Big) ( x)
	\qquad (x \in \T_r).
\end{equation}
By Lemma~\ref{lem:GT3},
there is a constant $C>0$ such that
\begin{align}
\label{eq:GG2Gstar-bis}
(\Tbr \star \big(\Tbr\big)^2 \star \Tbr) (x)
&\le (\taub * \taub^2 * \taub) (x) + C \frac{ \chib^2 }{r^d}
	+ C \bigg( \frac{ \chib^2 }{r^d} \bigg)^2 .
\end{align}
It follows from $\taub(x) \le \tau_{\beta_c}(x) \lesssim \nnnorm x^{-(d-2)}$
and a standard
convolution estimate \cite[Proposition~1.7]{HHS03} that
$(\taub * \taub^2 * \taub)(x) \lesssim \nnnorm x^{-(d-4)}$ when $d >4$,
so this term can be made as small as desired
by taking $\norm x_\infty \ge c_5$ with $c_5$ large enough, uniformly in $\beta \le \beta_c$.
Also, for $\beta\le \beta_c - c_4r^{-d/2}$, by \eqref{eq:chibds} we have
$r^{-d/2}\chi(\beta) \lesssim c_4^{-1}$, and we can make this as small as
desired by taking $c_4$ large enough.
With \eqref{eq:GG2Gstar-bis}, this shows that the factor multiplying
$\chib/r^d$ in \eqref{eq:lower_bound_T2} can be made as small as desired by
taking $\norm x_\infty \ge c_5$ and $c_4, c_5$ large.
In particular, it can be made smaller than $\frac 1 2 c_T$,
which completes the proof.
\end{proof}

\section{Coupling and switching: Proof of Theorem~\ref{thm:coupling_bounds}}
\label{sec:currents}

In this section, we introduce a new coupling between Ising models on $\Tr$ and
on large finite subsets of $\Z^d$.
The coupling uses the random current representation (see \cite{Dumi20} or \cite[Chapter~2]{Pani24_thesis} for an introduction) and a percolation exploration process.
To derive useful bounds, we develop a new switching lemma between currents on $\T_r$ and on large
finite subsets of $\Z^d$, in Lemma~\ref{lemma: switching torus 1}.

\subsection{Random currents and switching lemma}

\begin{definition}
Let $\mathcal G=(V,E)$ be a finite graph.
A \emph{current} $\n$ on $\mathcal G$ is a function from $E$ to $\mathbb N_0=\lbrace 0,1,\ldots\rbrace$.
We write $\Omega_{\mathcal G}$ for the set of all currents on $\mathcal G$. The set of \emph{sources} of $\n$, denoted by $\sn$, is defined as
\begin{equation}
\sn = \Big\{ x \in V \mid \sum_{y:\: xy\in E}\n_{xy}\textup{ is odd}\Big\}.
\end{equation}
\end{definition}

It is often useful to interpret a current $\n$ as a multigraph on $V$ with $\n_{xy}$ distinct edges connecting $x$ and $y$ for all $xy\in E$.
A sourceless current then can be seen as the edge count of a multigraph
consisting of a union of loops. Adding sources to $\n$ amounts
to adding a collection of paths connecting the elements of $\sn$ pairwise.
In particular, a current $\n$ with sources $\sn=\{u,v\}$ can be seen as the edge count of a multigraph
consisting of a union of loops together with a path from $u$ to $v$.

For $A\subset V$, we write $\sigma_A=\prod_{x\in A}\sigma_x$.
It is well-known \cite{Aize82,Dumi20} that
the correlation functions of the Ising model (defined in \eqref{eq:def_Ising}) are related to currents via
\begin{equation}
\label{equation correlation rcr}
    \left\langle \sigma_A\right\rangle_{\beta}^{\mathcal G}
    =\dfrac{\sum_{\sn=A}w_\beta(\n)}{\sum_{\sn=\varnothing}w_\beta(\n)}
    =
    \frac{Z^{A}_{\mathcal G,\beta} }{Z^{\varnothing}_{\mathcal G,\beta} }
    ,
\end{equation}
where
\begin{equation}
\label{eq:wZ-def}
w_\beta(\n)=
\prod_{e \in E}\dfrac{\beta^{\n_{e}}}{\n_{e}!}
,
\qquad
Z^{A}_{\mathcal G,\beta} = \sum_{\sn=A}w_\beta(\n).
\end{equation}
For a set of currents $\mathcal E \subset \Omega_{\mathcal G}$, we also define
\begin{equation}
\label{eq:ZAG}
Z^{A}_{\mathcal G,\beta}[\mathcal E]= \sum_{\sn=A}w_\beta(\n)\1\{\n\in \mathcal E\}.
\end{equation}
Observe that, by definition, $Z_{\mathcal G,\beta}^A[\Omega_{\mathcal G}]=Z^A_{\mathcal G,\beta}$.

A current $\n$ induces a bond percolation configuration on $\mathcal G$ defined by $(\1_{\n_e>0})_{e\in E}$.
The connectivity properties of this percolation configuration are known
to play a crucial role in the analysis of the underlying Ising model.
This motivates the following notation.
For a current $\n$ on $\mathcal G$ and $x,y\in V$, we write $x\xleftrightarrow{\n\:}y$ if there exists a sequence of points $x_0=x,x_1,\ldots, x_k=y$ such that $x_ix_{i+1}\in E$ and $\n_{x_ix_{i+1}}>0$ for all $0\leq i < k$.

Given  two graphs $\mathcal G,\mathcal H$,
two vertex sets $A\subset \mathcal G$, $B\subset \mathcal H$,
and $\mathcal E\subset \Omega_{\mathcal G}\times \Omega_{\mathcal H}$, we write
\begin{equation}
\label{eq:ZABGH}
Z^{A,B}_{\mathcal G,\mathcal H,\beta}[\mathcal E]= \sum_{\substack{\sn_1=A\\\sn_2=B}}w_\beta(\n_1)w_\beta(\n_2)\1\{(\n_1,\n_2)\in \mathcal E\}.
\end{equation}
We often drop the subscript $\beta$ when it does not play a role.
For two sets $S_1,S_2$,
we write $S_1 \Delta S_2 = (S_1 \setminus S_2 ) \cup ( S_2 \setminus S_1)$ for their symmetric difference.
For two currents $\n$ and $\m$, we write $\n \le \m$ when $\n_{e}\le \m_{e}$ for each edge $e$,
and we write $\binom{\m}{\n} = \prod_{e\in E}\binom{\m_e}{\n_e}$ for the product of the binomial coefficients.
We will use the following minor extension of the switching lemma
of \cite{GHS70,Aize82} (see also \cite[Lemma~7]{Dumi20}).

\begin{lemma}[Switching lemma]
\label{lem:std_switch}
Let $\Gcal=(V,E)$ and $\mathcal H=(V',E')$ be finite graphs with $\mathcal H\subset \mathcal G$.
For a current $\n$ on $\mathcal G$, let $\n_{|\mathcal H}$ denote its restriction to edges of $\mathcal H$.
Then, for any $A, B\subset V$, $x,y\in V'$, and
$\Ecal \subset \Omega_\Gcal$,
\begin{align}
\label{eq:std_switch}
Z^{A,B}_{\mathcal G,\mathcal G} \Big[
	\n_1 + \n_2 \in \Ecal  ,\:
	 x \connect{(\n_1 + \n_2)_{|\mathcal H}\:} y
\Big]
&= Z^{A\Delta \{x,y\},B\Delta \{x,y\}}_{\mathcal G,\mathcal G} \Big[
	\n_1 + \n_2 \in \Ecal  ,\:
	 x \connect{(\n_1 + \n_2)_{|\mathcal H}\:} y
\Big] .
\end{align}
\end{lemma}

It is an immediate consequence of \eqref{eq:std_switch} that
\begin{equation}
\label{eq:switchbd}
Z^{A,B}_{\mathcal G,\mathcal G} \Big[
	 x \connect{(\n_1 + \n_2)_{|\mathcal H}\:} y
\Big]
\le
 Z^{A\Delta \{x,y\}}_{\mathcal G}
 Z^{B\Delta \{x,y\}}_{\mathcal G}  .
\end{equation}

\begin{proof}
By definition, $w(\n_1)w(\n_2) = w(\n_1 + \n_2) \binom{\n_1+\n_2}{\n_1}$.
By the change of variables   $\m = \n_1 + \n_2$ and $\n = \n_2$, the left-hand side of \eqref{eq:std_switch} can be written as
\begin{equation}
\sum_{ \del \m = A \Delta B }  w(\m) \1\{\m\in \mathcal E\}
\1\{ x \xleftrightarrow{\m_{|\mathcal H}\: } y \}
\sum_{\substack{ \n\in\Omega_{\mathcal G}, \:\n \le \m \\ \del \n = B }} \binom{\m}{\n}.
\end{equation}
We view $\m, \n$ as multigraphs $\Mcal, \Ncal$,
and write $\del \Ncal \subset V$ for the vertices with an odd number of incident edges in $\Ncal$.
From this viewpoint, the sum over $\n$ above exactly counts the number of subgraphs $\Ncal$ of $\Mcal$ with $\del \Ncal = B$.

On the event $x \connect{\m_{|\mathcal H}\:} y$, there is a path in $\Mcal$ from $x$ to $y$ that uses only edges in $E'$. We denote such a path by $\Kcal$.
The involution $\Ncal \mapsto \Ncal \Delta \Kcal$ then defines a bijection between subgraphs with source set $B$ and source set $B\Delta \{x,y\}$. Thus,
\begin{equation}
\1\{ x \xleftrightarrow{\m_{|\mathcal H}\: } y \}\sum_{\substack{ \n\in\Omega_{\mathcal G}, \:\n \le \m \\ \del \n = B }} \binom{\m}{\n}
= \1\{ x \xleftrightarrow{\m_{|\mathcal H}\: } y \}\sum_{\substack{ \n\in\Omega_{\mathcal G}, \:\n \le \m \\ \del \n = B\Delta \{x,y\}}} \binom{\m}{\n}.
\end{equation}
Undoing the change of variables produces the right-hand side of \eqref{eq:std_switch}.
\end{proof}

\subsection{A percolation exploration}
\label{subsection: percolation exploration}

We now consider bond percolation and introduce an exploration of the cluster of the origin that is particularly well-suited for the periodic setup to which we apply it. We follow the strategy initially introduced in \cite{BS96} and revisited in \cite{Papa06,HHII11}, with some modifications.
We fix $r\geq 3$,
and let $\mathcal G=(V,E)$ be a finite or infinite subgraph of $\mathbb Z^d$.
Let $\omega=(\omega_e)_{e\in E}$ be a bond
percolation configuration on $\mathcal G$. We
will define an exploration $\mathcal C(0)$ of the cluster of $0$ in $\omega$ that
contains at most one
representative for each vertex and edge of $\Tr$, under the equivalence that
$x \cong x'$ if $x-x'\in r\Z^d$ (for $x,x'\in \mathcal G$).
A simple example of the exploration is depicted in Figure~\ref{figure:percoalgo}.

\begin{figure}
\begin{center}
\includegraphics{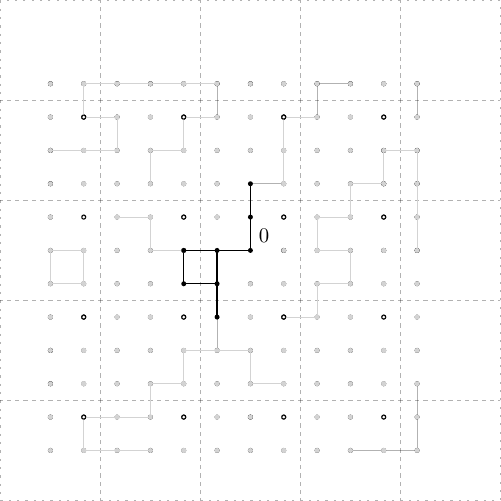}		
\caption{An illustration of the exploration algorithm with $r=3$ and $\mathcal G = \Z^2$. Dotted boxes represent copies of $\Lambda_3$.
The bond percolation configuration is shown in black and grey; the black edges are explored.
Vertices not explored by the algorithm are unfilled.
The graph $\mathcal C(0)$ consists of the black vertices and edges.
}
\label{figure:percoalgo}
\end{center}
\end{figure}

We use the following colouring scheme for vertices and edges of $\mathcal G$:
\begin{itemize}
	\item Vertices have three possible colourings: red
(being or about to be explored), black (fully explored), or grey (has a coloured representative). We write $\mathsf{R}$, $\mathsf{B}$, $\mathsf{G}$ for the corresponding evolving sets.
	\item Edges have two possible colourings: white (explored and closed) or black (explored and open). We let $\mathfrak{W}$, $\mathfrak{B}$ be the corresponding evolving sets.
\end{itemize}
The construction of $\mathcal C(0)$ is as follows.
We fix an (arbitrary) ordering of the vertices in $\mathcal{G}$.
\begin{enumerate}
	\item[\textbf{Step 0}] Initially $\mathsf{R}=\{0\}$, $\mathsf{B}=\varnothing$, $\mathsf{G}=\{0' \mid 0'\cong 0, \: 0'\neq 0\}$, $\mathfrak{W}=\mathfrak{B}=\varnothing$.
	\item[\textbf{Step 1}] Let $u$ be the earliest element of $\mathsf{R}$.
	\begin{enumerate}
		\item[$(i)$] Let $v$ be the earliest vertex adjacent to $u$ such that $v$ is not coloured grey
		and $uv$ is not coloured.
		\item[$(ii)$] If $\omega_{uv}=0$, add the edge $uv$ to $\mathfrak{W}$.
		\item[$(iii)$]  If $\omega_{uv}=1$, add $v$ to $\mathsf R$ (if it is not already in it) and all its representatives $v'\cong v$ ($v'\neq v)$ to $\mathsf{G}$, add $uv$ to $\mathfrak{B}$.
		\item[$(iv)$] Restart from $(i)$ until all the edges containing $u$ are either coloured or connected to a grey vertex.
		Then add $u$ to $\mathsf B$.
	\end{enumerate}
	\item[\textbf{Step 2}] Restart from Step 1 until $\mathsf R=\varnothing$.
\end{enumerate}
At the end of the exploration, we obtain $\mathsf B, \mathsf G, \mathfrak B, \mathfrak W$. Note that all edges that have been explored belong to $\mathfrak W \cup \mathfrak B$.
We define $\mathcal C(0)$ to be the graph $(\mathsf B,\mathfrak B)$.

\medskip
We note the following properties of $\mathcal C(0)$:
\begin{itemize}
\item
By construction,
$\mathcal C(0)$ contains at most one representative of each vertex and edge,
and $\mathcal C(0)$ is connected.
Any path in $\mathcal C(0)$ has at most $r^d$ vertices, so has length
at most $r^d-1$.

\item
Crucially, $\mathcal C(0)= (\mathsf B,\mathfrak B)$ is  measurable in terms of the state of the edges in $\mathfrak B \cup \mathfrak W$,
since
$\mathsf B = \{ 0 \} \cup \{ \text{vertices incident to $\mathfrak B$} \}$ by construction, and, for $(B,W)$ a valid output of the algorithm,
\begin{equation} \label{eq:BW}
\{ \mathfrak B = B,\:\mathfrak W = W \}
= \{\omega_e=1 \text{ for } e\in B,\ \omega_e = 0 \text{ for } e \in W \}.
\end{equation}
The $\subset$ inclusion is by definition of the algorithm.
On the other hand, if $(B,W)$ is a valid output of the algorithm and
$\omega$ is an element of the right-hand side,
then running the exploration on $\omega$ will produce $\mathfrak B = B, \:\mathfrak W = W$, so the $\supset$ inclusion follows.
\end{itemize}

\subsection{A coupled cluster exploration: Proof of \eqref{eq:upper_bound}}

We now apply the exploration introduced in Section~\ref{subsection: percolation exploration} to a well-chosen percolation configuration induced by random currents. This allows us to perform switching between a current on a large
finite subgraph $\mathcal G \subset \Z^d$ and a current on $\mathbb T_r$.
To proceed, we first extend the projection $\pi:\Z^d \to \T^d$ defined previously, and introduce the notion of a \emph{lift} from the torus to $\Z^d$.

\medskip \noindent {\bf Projections.}
Given a finite or infinite subgraph $\mathcal G = (V,E)$ of $\Z^d$,
we define its projection $\pi' \mathcal G = (\pi' V, \pi' E)$ to be the graph
on $\T^d$ with vertex set $\pi' V = \{\pi x: x\in V\}$ and edge set
$\pi' E = \{ \pi' e: e \in E\}$, where for an edge $e=\{x,y\}\in E$ we set
$\pi' e = \{\pi x, \pi y\}$.  Since the projection $\pi:\Z^d \to \T^d$ agrees
with the projection $\pi'$ applied to a graph consisting of a single vertex,
we drop the prime and write simply $\pi$ for the graph projection.  Thus $\pi$ can
act on vertices, on edges, and on graphs.
When $S$ is a subgraph of $\Zd$ whose vertex set does not contain any equivalent vertices, the restriction
$\pi_{|S}$ of $\pi$ to $S$ gives an isomorphism between the graphs $S$ and $\pi S$.
We denote its inverse by $\pi_{|S}^{-1}$.

\smallskip\noindent {\bf Lifts.}
Any function $F$ defined on the vertices and/or edges of the
torus $\T^d$ can be \emph{lifted} to the function $\hat F$ defined on the
vertices/edges of $\Z^d$ by $\hat F = F\circ \pi$.  In particular,
given a current $\n$ on $\T^d$ and a subgraph $\mathcal G$ of $\Z^d$, the \emph{lift}
$\hat \n$ of $\n$ is the current on $\mathcal G$  defined by
\begin{align}
\label{eq:nhatdef}
\hat\n = \n \circ \pi, \quad \text{i.e.,}\quad
\hat \n(e) = \n (\pi e)	 	\qquad (e\in \Gcal).
\end{align}

\smallskip

Let $\m$ be a current on $\Gcal$ and $\n$ be a current on $\Tr$.
The current $\m+\hat \n$ induces a percolation configuration
$(\1_{(\m + \hat \n)_e>0})_{e\in E}$
on $\mathcal G$.
We run the exploration of Section~\ref{subsection: percolation exploration} on this configuration and produce $(\mathsf B,\mathsf G,\mathfrak B,\mathfrak W)$ and $\mathcal C(0)=(\mathsf B, \mathfrak B)$.

\begin{Lem}[Properties of the exploration]
\label{lem:properties of the exploration}
Let $d \ge 1$, $r \ge 3$, $x\in\T_r$, and fix
a finite graph $\Gcal \supset \Lambda_{2r^d}$.
Let $\m\in \Omega_{\mathcal G}$,
and let $\n\in \Omega_{\mathbb T_r}$ with $\sn=\{0,x\}$.
Then the exploration $(\mathsf B,\mathsf G,\mathfrak B,\mathfrak W)$ of
$(\1_{(\m + \hat \n)_e>0})_{e\in E}$
satisfies the following properties:
\begin{enumerate}[label=(\roman*)]
\item
    $(\m+\hat \n)_e=0$ for all $e\in \mathfrak W$.
\item
If $e\in \mathfrak W$ then $\hat \n_{e'}=0$ for every edge $e'$
in $\mathcal G$ with $\pi {e'} =\pi e$.
\item
There exists exactly one representative of $x$, which we denote by $x'$, such that $x'\in \mathcal C(0)$.
\end{enumerate}
\end{Lem}

\begin{proof}
The first two points follow directly from the algorithm of the exploration and the definition of $\hat \n$.
The uniqueness of $x'$ in $(iii)$ has been noted already in Section~\ref{subsection: percolation exploration}.

For existence, since $\sn = \{0,x\}$, there is an open path in $\n$ from 0 to $x$, which we label as $(y_0, y_1, \dots, y_k)$ with $y_0 = 0$, $y_k = x$.
We claim that, for any $j \in [0,k]$, there exists $y_j' \in \Gcal$ such that $y_j' \cong y_j$ and $y_j' \in \Ccal(0)$.
Then, we can take $x' = y_k'$.
We prove the claim by induction on $j$. The $j=0$ case is trivial as we can take $y_0' = 0$.

Assume the claim for some $j < k$.
Observe that $y_j' \in \Ccal(0) \subset  \Lambda_{2r^d -1 } $,
because a path joining $0$ to $y_j'$ has length at most $r^d-1$.
Let $z = y_j' + (y_{j+1} - y_j) \in \Lambda_{2r^d} \subset  \Gcal$, so $z \cong y_{j+1}$.
If $z \in \Ccal(0)$, we take $y_{j+1}' = z$.
Otherwise, $z$ must be grey because $\hat \n_{y_j' z } = \n_{y_j y_{j+1}} > 0$. This only happens when another representative $z'$ of $z$ is explored before $z$.
In that case, we take $y_{j+1}' = z'$.  This completes the inductive step
and the proof.
\end{proof}

We use the exploration
to switch sources between currents on $\Gcal$ and $\Tr$,
via the following lemma.

\begin{lemma}\label{lemma: switching torus 1}
For any $C\subset \Gcal$, $A \subset \Tr$, and $z\in \Gcal$,
\begin{equation}
Z^{C,A}_{\mathcal G,\mathbb T_r}[z \in \mathcal C(0)]
= Z^{C\Delta\{0,z\}, A \Delta \{0 ,\pi z\}}_{\mathcal G,\mathbb T_r}[z \in \mathcal C(0)].
\end{equation}
In particular, for $x\in \Tr$, $C=\varnothing$, $A = \{ 0,x\}$, and $z = x'$ with $\pi z = x$,
\begin{equation}\label{eq:switchingtorus2}
Z^{\varnothing,0x}_{\mathcal G,\mathbb T_r}[x'\in \mathcal C(0)]
= Z^{0x',\varnothing}_{\mathcal G,\mathbb T_r}[x'\in \mathcal C(0)].
\end{equation}
\end{lemma}

\begin{proof}
Let $\m$ denote the current on $\mathcal G$ and $\n$ denote the current on $\T_r$.
Recall that $\mathcal C(0)$ consists of black vertices and edges.
By conditioning on $(\mathfrak B,\mathfrak W)$, we have
\begin{align}
Z^{C,A}_{\mathcal G,\mathbb T_r}[z \in \mathcal C(0)]
= \sum_{B\ni z,W} Z^{C,A}_{\mathcal G,\mathbb T_r}[(\mathfrak B,\mathfrak W)=(B,W)],
\end{align}
where $B\ni z$ means that $z$ is
incident to some edge in
$B$.
Henceforth, we write $\mathcal E_{(B,W)}$ for the event
$\{ (\mathfrak B,\mathfrak W)=(B,W) \}$, which is measurable in $( \m + \hat{\n} )_{|(B\cup W)}$, as indicated at \eqref{eq:BW}.
We will switch sources in $B\cup W$ (which we view as a graph).

Given an admissible pair $(B,W)$ (i.e.,
realisable by the exploration and such that $B\ni z$),  we make the decomposition
\begin{alignat}2
\m &= \m_1 + \m_2,   \qquad &&\m_1(e) = \m(e) \1\{e\in B\cup W\}, \\
\n &= \n_1+\n_2,  	\qquad &&~\n_1(e) = \n(e) \1\{e\in\pi(B \cup W)\} .
\end{alignat}
Then $w(\m) = w(\m_1)w(\m_2)$ and $w(\n) = w(\n_1)w(\n_2)$.
We also define $\tilde \n_1 = \hat \n_{|B\cup W}$.
Since no edges connect to grey vertices in the exploration, the
projection map $\pi$ restricts to
an isomorphism $\pi_{|B\cup W}$  between the graphs
$B\cup W$ and $\pi(B\cup W)$, and
\begin{equation} \label{eq:n1_lift}
\tilde \n_1  = \n_1 \circ \pi_{|B\cup W} ,
    \quad
    w(\tilde\n_1) = w(\n_1),
    \quad
    \pi_{|B\cup W}(\partial\tilde \n_1) = \sn_1.
\end{equation}

By conditioning on $\m_2$, $\n_2$,
\begin{align}
\label{eq:ZTswitch}
&Z^{C,A}_{\mathcal G,\mathbb T_r}[\mathcal E_{(B,W)}\cap \{z\in B\}]
\nnb
& \qquad\qquad\qquad = \sum_{\m_2, \n_2} w(\m_2) w(\n_2)
	\sum_{\substack{ \sm_1 = \sm_2\Delta C \\ \sn_1 = \sn_2 \Delta A}}
	w(\m_1) w(\tilde \n_1) \1_{\mathcal E_{(B,W)} }
\1 \{ 0 \connect{(\m_1+\tilde{\n}_1)_{|B} } z \} ,
\end{align}
where the connection $0 \connect{(\m_1+\tilde{\n}_1)_{|B} } z$ is present
because $B\ni z$.
In view of \eqref{eq:n1_lift}, we can change the sum over $\n_1$ to be a sum over $\tilde \n_1$.
Now, since $\tilde \n_1$ and $\m_1$ are both currents on $B \cup W$, we can use Lemma~\ref{lem:std_switch} (with $\mathcal H = B$ and $\mathcal E = \mathcal E_{(B,W)}$) to switch sources between them.
After the swiching, we use \eqref{eq:n1_lift} again to change the sum over $\tilde \n_1$ back to a sum over $\n_1$.
With $C_1=\sm_2\Delta C$ and $A_1=\sn_2 \Delta A$, this process is
represented by
\begin{align}
    \sum_{\substack{ \sm_1 = C_1 \\ \sn_1 = A_1}}
    =
    \sum_{\substack{ \sm_1 = C_1 \\ \partial\tilde\n_1 = \pi_{|B\cup W}^{-1} A_1}}
    =  \sum_{\substack{ \sm_1 = C_1 \Delta \{0,z\} \\ \partial\tilde\n_1 = (\pi_{|B\cup W}^{-1} A_1) \Delta \{0,  z\} }}
    =
    \sum_{\substack{ \sm_1 = C_1 \Delta \{0,z\} \\ \partial\n_1 =  A_1 \Delta \{0,\pi z\} }} ,
\end{align}
where the equalities of summation symbols mean that the corresponding
replacements are valid in \eqref{eq:ZTswitch}.
After reading the resulting adaptation of \eqref{eq:ZTswitch}
from right to left, we obtain
\begin{align}
Z^{C,A}_{\mathcal G,\mathbb T_r}[\mathcal E_{(B,W)} \cap \{z\in B\}]
= Z^{C\Delta\{0,z\},A\Delta \{0,\pi z\}}_{\mathcal G,\mathbb T_r}[\mathcal E_{(B,W)}\cap \{z\in B\}] .
\end{align}
Summation over all admissible pairs $(B,W)$ then yields the desired result.
\end{proof}

We can now prove the first bound of Theorem~\ref{thm:coupling_bounds}.

\begin{proof}[Proof of \eqref{eq:upper_bound}]
Let $\Gcal = \LambdaR$ for some $R \ge 2r^d$.
We introduce an auxiliary sourceless current $\m$ on $\LambdaR$.
Then \eqref{equation correlation rcr} and
Lemma~\ref{lem:properties of the exploration}(iii) give
\begin{equation}
	\taub^\Tr(x)
	=\frac{Z_{\mathbb T_r}^{0x}}{Z_{\mathbb T_r}^\varnothing}
	=\frac{Z_{\LambdaR,\mathbb T_r}^{\varnothing,0x}}{Z_{\LambdaR,\mathbb T_r}^{\varnothing,\varnothing}}
	= \frac{ Z^{\varnothing, 0x}_{\LambdaR, \Tr}
	\Big[ \bigsqcup_{x'\cong x} \{ x'\in \mathcal C(0)\} \Big] }
	{Z^{\varnothing,\varnothing}_{\LambdaR,\Tr}}
	= \sum_{x'\in\Lambda_R: x'\cong x}
\frac{ Z^{\varnothing, 0x}_{\LambdaR, \Tr}[x'\in \Ccal(0) ] }
	{Z^{\varnothing,\varnothing}_{\LambdaR,\Tr}  } .
\end{equation}
Using \eqref{eq:switchingtorus2}, and then dropping the event $\{x'\in \Ccal(0)\}$,
we get
\begin{equation} \label{eq:torus_identity}
\taub^\Tr(x)
= \sum_{x'\in\Lambda_R: x'\cong x} \frac{ Z^{0x', \varnothing}_{\LambdaR, \Tr}[x'\in \Ccal(0) ] }
	{Z^{\varnothing,\varnothing}_{\LambdaR,\Tr}  }
\le
\sum_{x'\in\Lambda_R: x'\cong x}
	\frac{ Z_\LambdaR^{0x'} Z_\Tr^\varnothing } { Z_\LambdaR^\varnothing Z_\Tr^\varnothing }
=
\sum_{x'\in\Lambda_R: x'\cong x}\taub^\LambdaR (0,x').
\end{equation}
Then \eqref{eq:upper_bound} follows from the fact that $\taub^\LambdaR(0,x') \le \taub(x')$
by the second Griffiths inequality.
\end{proof}

\subsection{Proof of \eqref{eq:lower_bound_diagram}}

Let $R \ge 2r^d$. We define
\begin{equation}
    \Tsum_\beta^{(r,R)}(x) = \sum_{x'\in\Lambda_R:x'\cong x} \tau_\beta^{\Lambda_R}(0,x')
    \qquad (x\in\T_r).
\end{equation}
Then, by the first equality of \eqref{eq:torus_identity},
\begin{equation} \label{eq:lower_bound_pf0}
\Tsum_\beta^{(r,R)}(x) - \tau^{\mathbb T_r}_\beta(x)
= \sum_{x'\in\Lambda_R: x'\cong x}
\frac{Z^{0x',\varnothing}_{\LambdaR,\mathbb T_r}[x'\notin \mathcal C(0)]}{Z^{\varnothing,\varnothing}_{\LambdaR,\mathbb T_r}}
.
\end{equation}
The inequality \eqref{eq:lower_bound_diagram} is an immediate consequence of the following lemma, by taking the limit $R\to \infty$.

\begin{lemma} \label{lem:lower_bound_R}
Let $d\ge 1$ and $R \ge 2r^d$.
For any $x \in \Tr$, we have
\begin{equation}\label{eq:lower_bound_R}
\Tsum_\beta^{(r,R)}(x) - \tau_\beta^{\T_r}(x)
\le \sum_{ x' \cong x } \sum_{y} \sum_{ \substack{y' \cong y \\ y' \ne y} } \sum_{z } \tau^\LambdaR_\beta(0,z) \tau^\LambdaR_\beta(z, y) \tau^\LambdaR_\beta(z,y') \tau^\Tr_\beta(\pi y' -\pi z) \tau^\LambdaR_\beta(y, x'),
\end{equation}
where $x',y,y',z$ range over $\LambdaR$ with the indicated restrictions.
\end{lemma}

The proof proceeds by bounding the terms on the right-hand side of \eqref{eq:lower_bound_pf0}.  These terms involve a current $\m$ on $\Lambda_R$
with $\sm = \{0,x'\}$, so
there is an open path in $\m$ from $0$ to $x'$.
When $x' \not \in \mathcal C(0)$, every such path must contain a grey site.
(Note that $x'$ might be uncoloured, because the sources do not guarantee that there is an $x''\in \mathcal C(0)$ with $x'' \cong x$.)
We let $y\in \LambdaR$ denote the last grey site on one of the paths from $0$ to $x'$,
so that $y \connect{\m\:} x'$ using only edges in $ (\mathfrak B\cup \mathfrak W)^c$.
Also, since $y$ is grey, there exists another representative $y' \cong y$, $y'\ne y$ such that $y' \in \mathcal C(0)$.
This proves the inclusion
\begin{align} \label{eq:yy'}
\{ x' \not\in \mathcal C(0) \} \subset
\bigcup_{y\in \LambdaR} \bigcup_{ \substack{ y' \cong y \\ y' \ne y } }
\{ y' \in \mathcal C(0) \} \cap \{ y \connect{\m\:} x' \text{ in } (\mathfrak B \cup \mathfrak W)^c \} .
\end{align}
As in the proof of Lemma~\ref{lemma: switching torus 1}, we regard an edge subset (such as $(\mathfrak B \cup \mathfrak W)^c$) as the graph whose edges are those in that subset and whose vertices
are the vertices comprising those edges.

\begin{lemma} \label{lem:3currents}
Let $ \mathcal A(x',y,y') = \{ y' \in \mathcal C(0) \} \cap \{ y \connect{\m\:} x' \text{ in } (\mathfrak B\cup \mathfrak W)^c \}$
denote the event on the right-hand side of \eqref{eq:yy'}.
Then,
\begin{align}
Z^{0x', \varnothing}_{\LambdaR, \Tr}[ \mathcal A(x',y,y') ]
\le Z^{0y,\varnothing}_{\LambdaR, \Tr}[ y'\in \Ccal(0) ] \taub^\LambdaR(y,x').
\end{align}
\end{lemma}

\begin{proof}
The left-hand side of the desired inequality involves two currents,
but we will need a third for the two-point function on the right-hand side.
For simplicity, we write $\Acal = \mathcal A(x',y,y')$.
Since the connection joining $y$ and $x'$
is a connection outside of $\mathfrak B \cup \mathfrak W$, we can use it to switch sources with a third independent current without affecting $\mathfrak B \cup \mathfrak W$, as follows.
We again condition on the event
$\mathcal E_{(B,W)}=\{(\mathfrak B,\mathfrak W)=(B,W)\}$ and make the decomposition
\begin{align}
\m = \m_1 + \m_2, \qquad \m_1(e) = \m(e) \1\{e\in B\cup W\}.
\end{align}
Let $\l$ be an independent current on $(B\cup W)^c$ with $\del \l  = \varnothing$.
By conditioning on $\m_1$, we have
\begin{align}
\label{eq:break_tail_pf1}
&Z^{0x', \varnothing}_{\LambdaR, \Tr}[ \Acal ]
= \sum_{B,W} Z^{0x', \varnothing}_{\LambdaR, \Tr}[ \Acal
\cap \mathcal E_{(B,W)}
]	\nl
&\quad = \sum_{B\ni y',W} Z^{0x', \varnothing}_{\LambdaR, \Tr}[
\mathcal E_{(B,W)} \cap
\{	y \connect{\m\:} x' \text{ in } (B \cup W)^c \} ] 	
\frac{Z^\varnothing_{(B\cup W)^c}}{Z^\varnothing_{(B\cup W)^c}}
\nl
&\quad = \sum_{B \ni y', W} \sum_{ \substack{ \m_1 \\ \sn=\varnothing} }
	 w(\m_1) w(\n)    \1 _{
\mathcal E_{(B,W)}}
	\!\!\!\!\!\!
    \sum_{ \substack{ \del \l = \varnothing \\ \del \m_2 = \del \m_1 \Delta\{ 0,x'\} } }
	w(\l) w(\m_2) \1\{ y \connect{\m_2\:} x' \}
\frac{1}{Z^\varnothing_{(B\cup W)^c}}.
\end{align}
We then switch sources between $\l$ and $\m_2$ on $\mathcal G= \mathcal H = (B \cup W)^c$.
Since $\{ y \connect{\m_2\:} x' \} \subset \{ y \connect{\l + \m_2\:} x' \}$,
from \eqref{eq:switchbd} we get
\begin{equation}
\label{eq:ell-m-bd}
\sum_{ \substack{ \del \l = \varnothing \\ \del \m_2 = \del \m_1 \Delta \{0,x'\}} }
	w(\l) w(\m_2) \1\{ y \connect{\m_2\:} x' \}
\le \sum_{  \del \m_2 = \del \m_1 \Delta \{0,y\} } w(\m_2)
	\sum_{  \del \l =  \{y,x'\}  } w(\l)  .
\end{equation}
After insertion of \eqref{eq:ell-m-bd} into \eqref{eq:break_tail_pf1},
the sum over $\l$, together with $1 / Z^\varnothing_{(B\cup W)^c}$, creates a two-point function on $(B\cup W)^c$ which is bounded by $\taub^{\Lambda_R}(y,x')$ by a Griffiths inequality.
The remaining factors reassemble
to provide the bound
\begin{equation}
 Z^{0x', \varnothing}_{\LambdaR, \Tr}[ \Acal ]
\le \taub^{\Lambda_R}(y,x')
\sum_{B \ni y', W} Z^{0y,\varnothing}_{\LambdaR,\Tr}[
	\Ecal_{(B,W)}]
= \taub^{\Lambda_R}(y,x') Z^{0y,\varnothing}_{\LambdaR,\Tr}[ y' \in \Ccal(0) ] ,
\end{equation}
as desired.
\end{proof}

In the proof of Lemma~\ref{lem:lower_bound_R}, we use the notion of
the backbone of a current \cite[Section 4]{AF86} (more on the backbone can be found in \cite{ADTW19,Pani24,Pani24_thesis}).
A current $\m$ on a finite graph $\mathcal G$ with sources $\sm=\{u,v\}$ must have at least one path from $u$ to $v$ in the percolation configuration induced by $\m$. The backbone of $\m$ is an exploration of one of these paths. We fix some arbitrary ordering $\prec$ of the edges of $\Gcal$.

\begin{definition}[The backbone of a current]
\label{def:backbone}
Let $\m\in \Omega_{\mathcal G}$ with $\sm=\lbrace u,v\rbrace$. The \emph{backbone}
of $\m$, denoted $\Gamma_{\m}$, is the unique oriented and edge self-avoiding path
from $u$ to $v$, supported on edges $e$ with $\m_{e}$ odd, that is minimal for $\prec$.
The backbone $\Gamma_{\m}=\{x_ix_{i+1}:  0\leq i < k\}$ is obtained via the following exploration process:
\begin{enumerate}
    \item[$(i)$] Let $x_0=u$. The first edge $x_0x_1$ of $\Gamma_{\m}$ is the earliest of all the edges $e$ emerging from $x_0$ with $\m_e$ odd.
    Thus, all edges $e\ni x_0$ satisfying $e\prec x_0x_1$ and $e\neq x_0x_1$ are explored and have $\m_e$ even.
    \item[$(ii)$] Inductively, each edge $x_ix_{i+1}$, $i\ge 1$, is the earliest of all edges $e$ emerging from $x_i$ that have not been explored previously, and for which $\m_e$ is odd.
    \item[$(iii)$] The exploration stops when it reaches a vertex from which no more odd, non-explored edges are available. This happens at $x_k=v$.
\end{enumerate}
Let $\overline{\Gamma}_{\m}$ denote the set of explored edges, i.e. $\Gamma_\m$ together with all explored even edges.
\end{definition}

By definition, $\overline \Gamma_\m$ is determined by $\Gamma_\m$ and the ordering $\prec$,
and similarly $\Gamma_\m$ can be determined from $\overline\Gamma_\m$ by selecting
the odd edges in $\overline\Gamma_\m$ according to $\prec$.
Also, the current $\m\setminus \overline \Gamma_\m$, defined to be the restriction of $\m$ to the complement $\overline{\Gamma}_\m^c$, is sourceless.

\begin{proof}[Proof of Lemma~\textup{\ref{lem:lower_bound_R}}]
By \eqref{eq:lower_bound_pf0}, \eqref{eq:yy'}, a union bound, and Lemma~\ref{lem:3currents},
\begin{equation} \label{eq:tree_pf0}
\Tsum_\beta^{(r,R)}(x) - \tau_\beta^{\T_r}(x)
\le \sum_{ x' \cong x } \sum_{y} \sum_{ \substack{y' \cong y \\ y' \ne y} }
	\frac{ Z^{0y,\varnothing}_{\LambdaR,\Tr}[y'\in \Ccal(0)] }
	{Z^{\varnothing,\varnothing}_{\LambdaR,\Tr}}
	\taub^\LambdaR(y,x') .
\end{equation}
The connections used to bound the above right-hand side are depicted
in Figure~\ref{figure:diagrammatic_bound}.
We first apply Lemma~\ref{lemma: switching torus 1} to switch on the event $y' \in \mathcal C(0)$.
Using Lemma~\ref{lemma: switching torus 1} with $z=y'$, $C=\{0,y\}$, and $A=\varnothing$, we get
\begin{equation} \label{eq:switching yy'}
Z^{0y,\varnothing}_{\LambdaR, \Tr}[ y'\in \Ccal(0) ]
= Z^{yy',0\pi y'}_{\LambdaR, \Tr}[ y'\in \Ccal(0) ] .
\end{equation}

\begin{figure}[h]
\begin{center}
\includegraphics{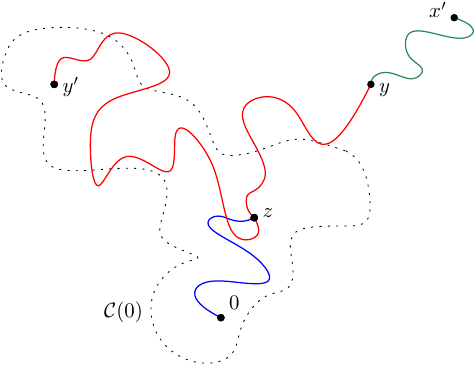}
\caption{An illustration of the connections relevant to
the right-hand side of
\eqref{eq:tree_pf0}.
The region delimited by the dotted line corresponds to the exploration $\mathcal C(0)$. The connection from $y$ to $x'$ (in green) does not involve any edge of $(\mathfrak B,\mathfrak W)$. The point $z$ is defined by \eqref{eq:prooflastswitching1}. Following the notation below \eqref{eq:prooflastswitching1}: the red path corresponds to $\Gamma_\m$ after the switching performed in \eqref{eq:switching yy'}, the blue path corresponds to a connection in $\mathcal C_{\m\setminus \overline\Gamma_\m+\hat \n}(0)$. The additional torus two-point function appearing in the right-hand side of \eqref{eq:lower_bound_R} is produced by the switching of
\eqref{eq:prooflastswitching3}.
}
\label{figure:diagrammatic_bound}
\end{center}
\end{figure}

Since the current $\m$ on $\Lambda_R$ now has sources $\sm = \{y,y'\}$, we can explore the backbone $\Gamma_\m$ of $\m$ from $y'$ to $y$.
Let $\overline \Gamma_\m$ be the set of edges explored during the formation of $\Gamma_\m$.
Let $\mathcal C_{ \m \setminus \overline{\Gamma}_\m+\hat \n}(0)$ denote
the result of the exploration process of Section~\ref{subsection: percolation exploration}
performed on the current $ \m\setminus \overline{\Gamma}_\m +\hat \n$
(on the graph $\LambdaR \setminus \overline \Gamma_\m$).
We claim that
\begin{equation} \label{eq:prooflastswitching1}
\{ y'\in \mathcal C(0) \}
\subset \bigcup_{z\in \LambdaR} \{z\connect{\overline{\Gamma}_\m \:}y'\}\cap \{z\in \mathcal C_{\m\setminus \overline{\Gamma}_\m +\hat \n}(0) \}.
\end{equation}
Indeed, if $\mathcal C_{\m\setminus \overline{\Gamma}_\m +\hat \n}(0) = \Ccal(0)$, then \eqref{eq:prooflastswitching1} holds with $z = y'$.
If $\mathcal C_{\m\setminus \overline{\Gamma}_\m +\hat \n}(0) \ne \Ccal(0)$,
then there is a first edge $zv \in \overline\Gamma_\m$ used in the construction of $\mathcal C(0)$ that makes the two exploration processes different.
For this edge, we must
have $z\in \mathcal C_{\m\setminus \overline{\Gamma}_\m +\hat \n}(0)$ and $\m_{zv} > 0$.
Since one of $z,v$ must be in $\Gamma_\m$, we see that $z \connect { \overline \Gamma_\m \:} y'$ also.  This proves \eqref{eq:prooflastswitching1}.

Conditional on $\Gamma_\m  = \gamma $, we set $\bar\gamma = \overline \Gamma_\m$ and
decompose $\m$ as
\begin{align}
\m = \k + \l,
\qquad \k(e) = \m(e) \1\{ e\in \bar{\gamma}\},
\end{align}
so $\k, \l$ are currents on $\bar{\gamma}$, $\LambdaR \setminus \bar{\gamma}$ respectively.
Note that $w(\m) = w(\k) w(\l)$, $\partial \k=\{y,y'\}$, $\partial \l=\varnothing$, and $\Gamma_\m = \Gamma_\k$.
Using \eqref{eq:prooflastswitching1},
\begin{equation}
Z^{yy',0\pi y'}_{\LambdaR, \Tr}[ y'\in \Ccal(0) ]
\le \sum_{\substack{z\in \LambdaR}}  \sum_\gamma
	\sum_{ \del \k = \{y, y'\} } w(\k) \1 \{\Gamma_\k=\gamma\}
	\1 \{ z \connect{ \overline{\Gamma}_\k\:} y' \}
	Z^{\varnothing,0\pi y'}_{\LambdaR\setminus \bar{\gamma},\mathbb T_r}[z\in \mathcal C_ {\l + \hat \n}(0)].
\end{equation}
By Lemma~\ref{lemma: switching torus 1} with
$\Gcal = \LambdaR \setminus \bar \gamma$,
\begin{align} \label{eq:prooflastswitching3}
Z^{\varnothing,0\pi y'}_{\LambdaR\setminus \bar{\gamma},\mathbb T_r}[z\in \mathcal C_ {\l + \hat \n}(0)]
&=
Z^{0z,\pi z\pi y'}_{\LambdaR\setminus \bar{\gamma},\mathbb T_r}[z\in \mathcal C_ {\l + \hat \n}(0)]
\nnb &\leq
Z_{\LambdaR \setminus \bar \gamma}^{0z} Z^{\pi z\pi y'}_{\mathbb T_r}
=
Z_{\LambdaR \setminus \bar \gamma}^{\varnothing}
	\taub^{\LambdaR \setminus \bar \gamma}(0,z)
Z^{\pi z\pi y'}_{\mathbb T_r}
.
\end{align}
Then we use $\taub^{\LambdaR \setminus \bar \gamma} \le \taub^\LambdaR$
to obtain
\begin{align}
Z^{yy',0\pi y'}_{\LambdaR, \Tr}[ y'\in \Ccal(0) ]
&\le \sum_{\substack{z\in \LambdaR}} Z^{\pi z\pi y'}_{\mathbb T_r}
	\sum_\gamma
	\sum_{ \del \k = \{y, y'\} } w(\k) \1 \{\Gamma_\k=\gamma\}
	\1 \{ z \connect{ \overline{\Gamma}_\k\:} y' \}
	Z_{\LambdaR \setminus \bar \gamma}^{\varnothing}
	\taub^{\LambdaR}(0,z)	\nl
&= \sum_{\substack{z\in \LambdaR}} Z^{\pi z\pi y'}_{\mathbb T_r}
	\sum_\gamma
	Z^{yy'}_\LambdaR [ \Gamma_\m = \gamma,\,
		z \connect{ \overline{\Gamma}_\m\:} y' ]
	\taub^{\LambdaR}(0,z)	 \nl
&= \sum_{\substack{z\in \LambdaR}} Z^{\pi z\pi y'}_{\mathbb T_r}
	Z^{yy'}_\LambdaR [ z \connect{ \overline{\Gamma}_\m\:} y' ]
	\taub^{\LambdaR}(0,z)	 .\label{eq:last eq3.5}
\end{align}
For the middle factor on the right-hand side,
we introduce an auxiliary sourceless current $\m'$ on $\LambdaR$,
and use
Lemma~\ref{lem:std_switch}
with $\Gcal = \Hcal = \LambdaR$:
\begin{align}
Z^{yy'}_\LambdaR [ z \connect{ \overline{\Gamma}_\m\:} y' ]
\le \frac{ Z^{yy',\varnothing}_{\LambdaR,\LambdaR }
	[ z \connect{\m+\m'} y'] }
	{Z^{\varnothing}_{\LambdaR}}
=
    \frac{ Z^{yz,zy'}_{\LambdaR,\LambdaR }	 }
	{Z^{\varnothing}_{\LambdaR}}
= Z^{yz}_{\LambdaR}\tau_\beta^{\LambdaR}(z,y').
\end{align}
Inserting this into \eqref{eq:last eq3.5} yields
\begin{align}
\frac{ Z^{yy',0\pi y'}_{\LambdaR, \Tr}[ y'\in \Ccal(0) ] }
	{Z^{\varnothing,\varnothing}_{\LambdaR,\Tr}}
&\le \sum_{\substack{z\in \LambdaR}}
	\frac{ Z^{\pi z\pi y'}_{\mathbb T_r} } { Z^\varnothing_\Tr }
	\frac{ Z^{yz}_{\LambdaR} } { Z^\varnothing_\LambdaR }
	\tau_\beta^{\LambdaR}(z,y')
	\taub^{\LambdaR}(0,z)	\nl
&= \sum_{\substack{z\in \LambdaR}}
	\taub^\Tr(\pi y' - \pi z) \taub^\LambdaR( z, y)
	\tau_\beta^{\LambdaR}(z,y')
	\taub^{\LambdaR}(0,z)	.
\end{align}
With \eqref{eq:tree_pf0} and \eqref{eq:switching yy'}, this concludes the proof.
\end{proof}

\appendix

\section{Torus convolution estimates}
\label{sec:convolution}

We present two elementary convolution lemmas.
The first is a special case of \cite[Lemma~3.4]{MS23}.
We include the simple proof for completeness, and as an opportunity to
correct an omission in the proof of \cite[Lemma~3.4]{MS23}
(repeated in \cite[Lemma~4.2]{MPS23}) which however does not change the conclusion of those lemmas.

\begin{lemma}
\label{lem:unifmassint}
Let $d \ge 1$, $r \ge 1$,
$a >0$ and $\mu >0$.  There
is a constant $C = C(d,a)>0$
such that
\begin{align}
	\sum\limits_{u \in\Z^d : u \neq 0} \frac{1}{\|x + r u\|_\infty^{d-a}}e^{- \mu \|x+ru\|_\infty}
	&\leq
    C
    e^{-\frac 14 \mu r}
    \frac{1}{\mu^a r^d}
     \qquad
     (x\in \Lambda_r)
     .
\end{align}
\end{lemma}

\begin{proof}
For all nonzero $u \in \Z^d$ and $r\geq 1$,
we have
$\norm {x + r u}_\infty \geq \frac12 \|ru\|_\infty$ by \eqref{eq:xulb}.
Hence,
\begin{align}
	\sum_{u \in \Z^d: u \neq 0}
    \frac{1}{\|x + r u\|_\infty^{d-a}}e^{-\mu \|x+ru\|_\infty}
	&\le 2^{d-a} \sum_{N = 1}^\infty \sum_{u:\|u\|_\infty =N}
    \frac{1}{\|ru\|_\infty^{d-a}}e^{-\frac 12 \mu \|ru\|_\infty}
    \nnb
&\lesssim r^{a-d} e^{-\frac 14 \mu r}
    \sum_{N = 1}^\infty  N^{a-1} e^{-\frac 14 \mu rN}
	.
\end{align}
We set $f(x) = x^{a-1} e^{-bx}$ for $x>0$, and $b = \frac 1 4 \mu r >0$,
so the goal is to prove that $\sum_{N=1}^\infty f(N) \lesssim b^{-a}$.
If $a \in (0, 1]$ then $f$ is a decreasing function, so the sum is bounded by $\int_0^\infty f(x) \mathrm dx = b^{-a} \Gamma(a)$.
If $a > 1$ then $f(x)$ is increasing for $x < x_0 = \frac{a-1}{b}$ and is decreasing for $x > x_0$. We set $N_0 = \lfloor x_0 \rfloor$ and bound the sum by
\begin{equation}
\bigg( \int_1^{N_0} f(x) \mathrm dx + f(N_0) \bigg) \1_{N_0 \ge 1}
+ f(N_0+1) +  \int_{N_0+1}^\infty f(x) \mathrm dx .
\end{equation}
The sum of the two integrals is again bounded by $\int_0^\infty f(x) \mathrm dx = b^{-a} \Gamma(a)$.
For the other terms,
if $N_0 = 0$ we use that $f(N_0+1) = f(1) = e^{-b} \le c_a b^{-a}$ for all $b>0$.
If $N_0 \ge 1$, we use
\begin{equation}
f(N_0), f(N_0+1) \le
\max_{x\ge0} f(x)
= \frac 1 {b^{a-1}} \max_{t\ge0} t^{a-1} e^{-t}
\le \frac{ a-1} {b^a} c_a,
\end{equation}
where the last inequality is due to $b \le a-1$ since $x_0 \ge N_0 \ge 1$.
This completes the proof.
\end{proof}

The next lemma is a minor extension of \cite[Lemma~3.6]{MS23}.
The definition of $\Tbr(x)$ is in \eqref{eq:Tsum}, $\star$ denotes
convolution on $\T_r$ and $*$ denotes convolution on $\Z^d$.

\begin{lemma}
\label{lem:GT3}
Let $d>4$.
There is a constant $C>0$ such that
for all $r\ge 3$, $\beta < \beta_c$, and $x\in \Tr$,
\begin{align}
\label{eq:plateau-ubA}
	\Tbr(x) &\leq \tau_\beta(x) + C\frac{\chi(\beta)}{r^d} , \\
\label{eq:GGstar}
    (\Tbr \star \Tbr )(x)
    &\le (\tau_\beta * \tau_\beta) (x) + C\frac{\chi(\beta)^2}{r^d}
    ,
    \\
\label{eq:GG2Gstar}
    (\big(\Tbr\big)^2 \star \Tbr \star \Tbr) (x)
    &\le (\taub^2 * \taub * \taub) (x) + C \frac{ \chib^2 }{r^d}
	+ C \bigg( \frac{ \chib^2 }{r^d} \bigg)^2
    .
\end{align}
\end{lemma}

\begin{proof}
We suppress the subscript $\beta$ and superscript $r$ in the notation,
and write $m=(\beta_c-\beta)^{1/2}$.  Then $m^{-2} \asymp \chi$ by \eqref{eq:chibds}.
For \eqref{eq:plateau-ubA}, we separate the $u = 0$ term
in the definition \eqref{eq:Tsum} of $T(x)$,
and then use \eqref{eq:DPub} and the $a=2$ case of Lemma~\ref{lem:unifmassint}, to obtain
\begin{align}
\label{eq:GammachiV}
	T(x) &\leq  \tau(x) + \sum_{u \neq 0}
	\frac{C_0}{\|x + r u\|_\infty^{d-2}}e^{-c_0 m \|x+ru\|_\infty}
	\leq \tau(x) + C\frac{\chi}{r^d}.
\end{align}

For \eqref{eq:GGstar}, we first observe that
\begin{align}
    (T \star T )(x)
    & =
    \sum_{y \in \Tr} \sum_{u,v \in \Z^d} \tau (y+ru) \tau (x-y + rv)
    \nnb & =
    \sum_{w\in\Z^d} \sum_{y \in \Tr} \sum_{v \in \Z^d} \tau (y-rv+rw) \tau (x-y + rv)
     =
    \sum_{w\in \Z^d} (\tau *\tau )(x-rw)
    ,
\label{eq:GamGam}
\end{align}
where in the second equality we replaced $u$  by $w-v$, and in the third we observe
that $y-rv$ ranges over $\Z^d$ under the indicated summations.
The $w=0$ term is $(\tau * \tau) (x)$.
For the remaining terms, we use the convolution estimate of \cite[Proposition~1.7]{HHS03} and \eqref{eq:DPub} to get
$(\tau *\tau)(z) \lesssim \nnnorm z ^{-(d-4)}\exp[-c_{0}m\|z\|_\infty]$.
We then apply Lemma~\ref{lem:unifmassint} with $a=4$ to see that
\begin{equation}
\label{eq:tautaubd}
    \sum_{w\neq 0} (\tau *\tau )(x-rw)
    \lesssim
    \frac{\chi^2}{r^d} ,
\end{equation}
which gives \eqref{eq:GGstar}.

For \eqref{eq:GG2Gstar},
by applying the inequality $T \le \tau \circ \pi  + C \chi / r^d$ from \eqref{eq:plateau-ubA}, we get three terms
by expanding the $T^2$ in $(T^2 \star T \star T)(x)$.
The first term is
\begin{align}
    ((\tau \circ \pi)^2 \star T \star T) (x)
    & =
    \sum_{y,z\in \T_r}\sum_{v,w\in\Z^d}
    \tau^2(y)\tau(z+rw-y-rv)\tau(x-z-rw)
    \nnb & =
    \sum_{y\in \T_r}\sum_{v\in\Z^d}
    \tau^2(y)(\tau*\tau)(x-y-rv)
    \nnb &
    \le \sum_{u\in \Zd} (\tau^2 * \tau * \tau)(x+ru) .
\end{align}
Since $d>4$, by \cite[Proposition~1.7]{HHS03},
$\tau^2*\tau*\tau$ obeys the same upper bound as the bound used
on $\tau*\tau$ in the proof of \eqref{eq:GGstar}, so we again
have a bound $\chi^2/r^d$ for the sum over nonzero $u$.
The cross term is
\begin{equation}
2C \frac{\chi}{r^d} ( (\tau \circ \pi) \star T \star T )(x)
\lesssim \frac{\chi}{r^d} \norm{ \tau }_{L^1(\Tr)} \norm{ T \star T }_\infty
\lesssim \frac{\chi}{r^d}\chi \Big( 1 + \frac{\chi^2}{r^d} \Big),
\end{equation}
by \eqref{eq:GGstar} together with the fact that $\|\tau*\tau\|_\infty \lesssim 1$ when $d>4$.
The last term is
\begin{equation}
\Big( C \frac{\chi}{r^d} \Big)^2 ( 1 \star T \star T )(x)
\lesssim \frac {\chi^2}{r^{2d}} \norm 1_\infty \norm T_{L^1(\Tr)}^2
= \frac {\chi^2}{r^{2d}} \chi^2,
\end{equation}
since $\norm T_{L^1(\Tr)} = \chi$.
This completes the proof.
\end{proof}

\section{Proof of Proposition~\ref{prop:u4}}
\label{sec:u4proof}

We now prove Proposition~\ref{prop:u4}.
Although this differential inequality is essentially the same as \cite[(12.118)]{FFS92},
the justification in \cite{FFS92} involves an unpublished work
of Fern\'andez, and to our knowledge a proof has not been published
elsewhere.  We therefore present a proof here.
We work on the torus and
to simplify the notation we drop the labels $\beta$ and $\mathbb T_r$.
We apply the \emph{dilution trick}, following the strategy of \cite{AF86}.

We begin with a reduction to the following proposition.
Although we only need the nearest-neighbour case, we
prove the result for a general Hamiltonian
\begin{equation}
    H (\sigma) = - \!\!\!\! \sum_{\{x,y\}\subset \Tr} J_{x,y}\sigma_x \sigma_y,
\end{equation}
where $J_{x,y} \ge 0$ is symmetric and translation invariant.
We denote the associated measure as $\langle \cdot\rangle$.
We write $|J|=\sum_{x\in \T_r}J_{0,x}$.  For the nearest-neighbour model, $|J|=2d$.
The weights defined in \eqref{eq:wZ-def} now become replaced by
\begin{equation}
\label{eq:wJ}
    w(\n)
    =
    \prod_{\{x,y\} \subset \Tr}\dfrac{(\beta J_{x,y})^{\n_{xy}}}{\n_{xy}!} .
\end{equation}
Recall that the torus bubble diagram is $\bubble(\beta)=\sum_{x\in\T_r}\tau_\beta(x)^2$.

\begin{proposition}
\label{prop:U4bd}
Let $d\ge 1$.
For all $\beta>0$, for all $r\geq 3$,
and for all $p\in[0,1]$,
\begin{equation} \label{eq:U4bd}
\sum_{x,y,z\in \mathbb T_r} |U_4(0,x,y,z)|
\geq
    2p \chi^2 \beta\frac{\partial \chi}{\partial \beta}
    - 32
    p^2\beta^2|J|^2\chi^4 \bubble.
\end{equation}
\end{proposition}

\begin{proof}
[Proof of Proposition~\textup{\ref{prop:u4}}]
The lower bound in \eqref{eq:U4bd} can be written as
\begin{equation}
    ap-bp^2 \quad \text{with} \quad
    a = 2 \chi^2 \beta\frac{\partial \chi}{\partial \beta} ,
    \quad
    b= 32 \beta^2 |J|^2\chi^4 \bubble.
\end{equation}
If $b \le \frac a2$ then we take $p=1$ and in this case we obtain
\begin{equation}
\label{eq:chderiv-lb-new1}
	\sum_{x,y,z\in \mathbb T_r}|U_4(0,x,y,z)|
    \geq
    a-b \ge
    \frac{a}{2} = \chi^2 \beta \frac{\partial \chi}{\partial \beta} .
\end{equation}
If instead $b> a/2$,
then $ap-bp^2$ is maximal at $p=\frac{a}{2b}\in (0,1)$, with maximum
$\frac{a^2}{4b}$.  We use this value of $p$ in \eqref{eq:U4bd}, and obtain
\begin{equation}
\label{eq:chderiv-lb-new2}
	\sum_{x,y,z\in \mathbb T_r}|U_4(0,x,y,z)|
    \geq
    \frac{a^2}{4b} =
    \frac 1 { 32
    |J|^2 \bubble}
    \Big(\frac{\partial \chi}{\partial \beta}\Big)^2.
\end{equation}
Then \eqref{eq:u4claim} follows by choosing the worst alternative in
\eqref{eq:chderiv-lb-new1}--\eqref{eq:chderiv-lb-new2}.
\end{proof}

It remains to prove Proposition~\ref{prop:U4bd}.
Given a vertex set $A\subset \T_r$,
we define a probability measure $\mathbf{P}^A$ on
the set
$\Omega$ of currents on $\T_r$ by
\begin{equation}
    \mathbf{P}^A[\n]
    =
    \1_{\sn=A}\frac{w(\n)}{Z^A}
    \qquad (\n \in \Omega),
\end{equation}
where $Z^A=\sum_{\sn=A}w(\n)$ is the normalisation constant.
Given two vertex sets $A_1,A_2$, we define the product measure
\begin{equation}
	\mathbf P^{A_1,A_2}=\mathbf P^{A_1}\otimes\mathbf P^{A_2},
\end{equation}
and we write $\mathbf E^{A_1,A_2}$ for expectation with respect
to $\mathbf P^{A_1,A_2}$.

We write $\mathcal P_2 = \mathcal P_2(\Tr)$ for the set of all unordered pairs
of vertices in $\T_r$.
Let $\mathbf C(u)$ denote the vertex cluster of $u\in\T_r$ in the
percolation configuration constructed from a double current $\n_1 + \n_2$, and
let $\overline{\mathbf C}(u)$ denote the set of
pairs $\{s,t\}\in \mathcal P_2$ such that either $s$ or $t$ belongs to
$\bfC(u)$.
Given $\theta \in [0,1]$ and $E\subset \mathcal P_2$, we introduce the (random) interaction
$(J^E_{u,v}(\theta))_{\{u,v\}\in \mathcal P_2}$
defined by
\begin{equation}
    J^E_{u,v}(\theta)
    =
    \begin{cases}
    \theta J_{u,v} & (\{u,v\} \in \overline{\mathbf C}(0)\cap E)
     \\
     J_{u,v} & (\{u,v\} \not\in \overline{\mathbf C}(0)\cap E) .
     \end{cases}
\end{equation}
We write
$\langle \cdot\rangle_\theta$ for the Ising expectation
with the random interaction $J^E(\theta)$.
Given a subset $G \subset \mathcal P_2$, we regard $G$ as a graph whose edges
consist of the elements of $G$ and whose vertices are the vertices contained in the
elements of $G$, and we write $\langle \cdot \rangle^G$ for the Ising expectation
on the graph $G$ with interaction $J$.
By definition, $\langle \cdot\rangle_1=\langle \cdot\rangle$ and $\langle \cdot \rangle_0=\langle \cdot \rangle^{(\overline{\mathbf C}(0)\cap E)^c}$.
According to our convention, the graph $(\overline{\mathbf C}(0)\cap E)^c$ consists
of the edges not in $\overline{\mathbf C}(0)\cap E$, together with the vertices comprising
those edges.

We start with a lower bound on $U_4$ expressed in terms of $T_1,T_2$ defined for $x,y,z \in \mathbb T_r$ and $E\subset \mathcal P_2$ by
\begin{align}
T_1(x,y,z,E)
&= \beta \sum_{\{u,v\} \in E} J_{u,v}
	\langle \sigma_0\sigma_x\rangle
	\mathbf P^{0x,\varnothing}[u\textup{ or }v\in \mathbf C(0)]
	\langle \sigma_y\sigma_z;\sigma_u\sigma_v\rangle,
\\
T_2(x,y,z,E)
&= \beta \sum_{\{u,v\} \in E} J_{u,v}
	\langle \sigma_0\sigma_x\rangle
\mathbf E^{0x,\varnothing} \bigg[\1\{u\textup{ or }v\in \mathbf C(0)\}
\int_0^1  \theta \;
	\partial_{\theta} \langle \sigma_y\sigma_z;\sigma_u\sigma_v\rangle_{\theta}
    \mathrm{d}\theta
    \bigg] .
\end{align}

\begin{lemma}
\label{lem:UTT}
For all $x,y,z \in \mathbb T_r$, and for any set $E\subset \mathcal P_2$,
\begin{equation}
\label{eq:UTT}
    \half |{U_4}(0,x,y,z)|\ge T_1(x,y,z,E) - T_2(x,y,z,E).
\end{equation}
\end{lemma}

\begin{proof}
By \cite[Proposition~5.1]{Aize82}, for all $x,y,z\in \mathbb T_r$,
\begin{align}
\label{eq:T1T2pf-1}
\half |U_4(0,x,y,z)|
&=  \langle \sigma_0\sigma_x\rangle \langle \sigma_y\sigma_z\rangle \mathbf P^{0x,yz}[\mathbf C(0)\cap \mathbf C(y)\neq \varnothing]  \nl
&=  \langle \sigma_0\sigma_x\rangle \langle \sigma_y\sigma_z\rangle \Big( 1 - \mathbf P^{0x,yz}[\mathbf C(0)\cap \mathbf C(y) = \varnothing] \Big)
.
\end{align}
We modify the definition \eqref{eq:ZABGH} by defining, for
$H \subset \mathcal P_2$ and for $A,B$ subsets of vertices in the graph $H$,
\begin{equation}
\label{eq:ZABGH-app}
Z^{A,B}_{H}[\;\cdot \;]= \sum_{\substack{\sn_1=A\\\sn_2=B}}w(\n_1)w(\n_2)\1\{(\n_1,\n_2)\in \cdot\},
\end{equation}
where the sums over $\n_1,\n_2$ are over currents on $H$ and the dot represents any set of pairs of currents.  When $H$ is all of $\mathcal P_2$, we omit the subscript $H$.
Then \eqref{eq:T1T2pf-1} can be written as
\begin{align} \label{eq:T1T2pf}
\half |U_4(0,x,y,z)|
&=  \langle \sigma_0\sigma_x\rangle \langle \sigma_y\sigma_z\rangle
	- \frac  { Z^{0x,yz}[\mathbf C(0)\cap \mathbf C(y) = \varnothing] } {Z^{\varnothing,\varnothing} } .
\end{align}
To prepare for an application of the switching lemma in
the proof of Lemma~\ref{lem:T1}, we wish to
arrange that one of the two currents be sourceless (as they are in
$T_1$ and $T_2$).

This can be achieved by the classical manoeuvre of
\emph{peeling off} the cluster of $0$, by conditioning on $\overline{\mathbf C}(0)$:
\begin{align}
	Z^{0x,yz}[\mathbf C(0)\cap \mathbf C(y)=\varnothing]
&=\sum_{H \subset \mathcal P_2}Z^{0x,\varnothing}_{H}[\overline{\mathbf C}(0)=H] Z^{\varnothing,yz}_{H^c}
	\nnb &=\sum_{H \subset \mathcal P_2}
Z^{0x,\varnothing}[\overline{\mathbf C}(0)=H]\langle \sigma_y\sigma_z\rangle^{H^c}
	=Z^{0x,\varnothing}  \mathbf E^{0x,\varnothing}[\langle \sigma_y\sigma_z\rangle^{\overline{\mathbf C}(0)^c}].
\end{align}
For every $E\subset \mathcal P_2$, it follows from the second Griffiths inequality that
\begin{equation}
    \langle \sigma_y\sigma_z\rangle^{\overline{\mathbf C}(0)^c}
\le \langle \sigma_y\sigma_z\rangle^{{(\overline{\mathbf C}(0)\cap E)^c}}=\langle \sigma_y\sigma_z\rangle_0.
\end{equation}
Therefore,
\begin{align}
Z^{0x,yz}[\mathbf C(0)\cap \mathbf C(y)=\varnothing]
&\le Z^{0x,\varnothing}   \mathbf E^{0x,\varnothing}[\langle \sigma_y\sigma_z\rangle_0] 		= Z^{\varnothing,\varnothing}  \langle \sigma_0 \sigma_x \rangle \mathbf E^{0x,\varnothing}[\langle \sigma_y\sigma_z\rangle_0] .
\end{align}
We insert this into \eqref{eq:T1T2pf} and obtain
\begin{align}
\half |U_4(0,x,y,z)|
\ge \langle \sigma_0\sigma_x\rangle\mathbf E^{0x,\varnothing}[\langle \sigma_y\sigma_z\rangle_1-\langle \sigma_y\sigma_z\rangle_0] .
\end{align}
By the Fundamental Theorem of Calculus,
\begin{align}
	\langle \sigma_y\sigma_z\rangle_1 - \langle \sigma_y\sigma_z\rangle_0	
    &=\int_0^1\partial_{\theta_1} \langle \sigma_y\sigma_z\rangle_{\theta_1}
    \mathrm{d}\theta_1
	\nnb &
    =\sum_{\{u,v\}\in E}\beta J_{u,v} \1\{u\textup{ or }v\in \mathbf C(0)\}
	\int_0^1\langle \sigma_y\sigma_z;\sigma_u\sigma_v\rangle_{\theta_1} \mathrm{d}\theta_1 .
\end{align}
The desired inequality \eqref{eq:UTT} then follows from
\begin{equation}
	\langle \sigma_y\sigma_z;\sigma_u\sigma_v\rangle_{\theta_1}
=
\langle \sigma_y\sigma_z;\sigma_u\sigma_v\rangle-\int^{1}_{\theta_1}\partial_{\theta}\langle \sigma_y\sigma_z;\sigma_u\sigma_v\rangle_{\theta }\mathrm{d}\theta
\end{equation}
and Fubini's theorem.
\end{proof}

In combination with Lemma~\ref{lem:UTT},
the bounds on $T_1$ and $T_2$ obtained in the next two lemmas prove Proposition~\ref{prop:U4bd}.
These bounds use the \emph{dilution trick}, which consists of picking $E=\omega$ with $\omega$ distributed
according to an independent Bernoulli$(p)$ bond percolation on $\mathcal P_2$,
and averaging over $E$.
We denote the law and expectation of $\omega$ by $\P_p$ and $\E_p$, respectively.
The dilution method is comparable to the second moment method (see e.g. \cite[Lemma~4.4]{AD21}) but it is more \emph{linear} in the sense that it allows for resummation. This
facilitates working with averaged quantities.

\begin{lemma}
\label{lem:T1}
For all $\beta >0$ and for all $p\in [0,1]$,
\begin{equation}
\sum_{x,y,z\in \mathbb T_r}\mathbb E_p[T_1(x,y,z,\omega)]
\geq p \chi^2 \beta\frac{\partial \chi}{\partial \beta}.
\end{equation}
\end{lemma}

\begin{proof}
By definition,
\begin{align}
T_1(x,y,z,E)  &=\beta \sum_{\{u,v\}\in E}J_{u,v}
\langle \sigma_0\sigma_x\rangle
\mathbf P^{0x,\varnothing}[u\textup{ or }v\in \mathbf C(0)]
\langle \sigma_y\sigma_z;\sigma_u\sigma_v\rangle
\\&=\frac{\beta}{2} \sum_{u,v\in\Tr}\1\{\{u,v\}\in E\}J_{u,v}
\langle \sigma_0\sigma_x\rangle
\mathbf P^{0x,\varnothing}[u\textup{ or }v\in \mathbf C(0)]
\langle \sigma_y\sigma_z;\sigma_u\sigma_v\rangle.
\end{align}
We first apply the switching lemma (Lemma~\ref{lem:std_switch})
to see that
\begin{equation}
	\langle \sigma_0\sigma_x\rangle \mathbf P^{0x,\varnothing}[u\textup{ or }v \in \mathbf C(0)]
	\geq \langle \sigma_0\sigma_x\rangle\mathbf P^{0x,\varnothing}[0\connect{\n_1+\n_2\:} u ]=\langle \sigma_0\sigma_u\rangle\langle \sigma_u\sigma_x\rangle.
\end{equation}
This enables us to take the percolation expectation and obtain
\begin{equation}
	\mathbb E_p[T_1(x,y,z,\omega)]
\geq
\frac{\beta p}{2}\sum_{ u,v\in \mathbb T_r}J_{u,v}\langle \sigma_0\sigma_u\rangle\langle\sigma_u\sigma_x\rangle\langle \sigma_y\sigma_z;\sigma_u\sigma_v\rangle
.
\end{equation}
Summation over the torus then gives
\begin{align}
\sum_{x,y,z\in \mathbb T_r} \mathbb E_p[T_1(x,y,z,\omega)]
\ge  \frac{\beta p}{2}\sum_{y,z,u,v\in \Tr}J_{u,v}\langle \sigma_0\sigma_u\rangle  \chi \langle \sigma_y\sigma_z;\sigma_u\sigma_v\rangle.
\end{align}
By the change of variables $z=z'+y$, $u=u'+y$, and $v=v'+y$,
the right-hand side becomes
\begin{equation}
\frac{\beta p}{2} \chi
	\sum_{y,z',u',v'\in \Tr}J_{u',v'}\langle \sigma_0\sigma_{u'+y}\rangle\langle \sigma_0\sigma_{z'};\sigma_{u'}\sigma_{v'}\rangle
=\beta p \chi ^2 \frac{ \del \chi }{ \del \beta },
\end{equation}
and the proof is complete.
\end{proof}

\begin{lemma}
For all $\beta >0$ and for all $p\in [0,1]$,
\begin{equation}\label{eq:bound T_2}
	\sum_{x,y,z\in \mathbb T_r}\mathbb E_p[T_2(x,y,z,\omega)]
    \leq
    16
    p^2 \beta^2 |J|^2\chi^4 \bubble.
\end{equation}
\end{lemma}

\begin{proof}
By definition,
\begin{equation}
T_2(x,y,z,E)
= \beta \sum_{\{u,v\}\in E} J_{u,v}
	\langle \sigma_0\sigma_x\rangle
\mathbf E^{0x,\varnothing} \bigg[\1\{u\textup{ or }v\in \mathbf C(0)\}
\int_0^1 \theta \;
	\partial_{\theta} \langle \sigma_y\sigma_z;\sigma_u\sigma_v\rangle_{\theta}
    \mathrm{d}\theta
    \bigg]
.
\end{equation}
An explicit computation gives
\begin{equation}
	\partial_{\theta}\langle \sigma_y\sigma_z;\sigma_u\sigma_v\rangle_{\theta}
= \beta \sum_{ \{s,t\} \in E }J_{s,t}\1\{s\textup{ or }t\in \mathbf C(0)\}
\langle \sigma_y\sigma_z;\sigma_u\sigma_v;\sigma_s\sigma_t\rangle_{\theta},
\end{equation}
where
\begin{equation}
\label{eq:6pt}
	\langle \sigma_y\sigma_z;\sigma_u\sigma_v;\sigma_s\sigma_t\rangle_{\theta}
= \langle \sigma_y\sigma_z;\sigma_u\sigma_v\sigma_s\sigma_t\rangle_{\theta}
	-\langle \sigma_u\sigma_v\rangle_{\theta} \langle \sigma_y\sigma_z;\sigma_s\sigma_t\rangle_{\theta}
	-\langle \sigma_s\sigma_t\rangle_{\theta} \langle \sigma_y\sigma_z;\sigma_u\sigma_v\rangle_{\theta}.
\end{equation}
When the edges $uv$ and $st$ are the same edge, the first term
on the right-hand side of \eqref{eq:6pt}
vanishes, as $\sigma_u \sigma_v \sigma_s \sigma_t = 1$. In this case, we have
$\langle \sigma_y\sigma_z;\sigma_u\sigma_v;\sigma_u\sigma_v\rangle_{\theta}
	\leq 0.$
Therefore,
\begin{align} \label{eq: simplification T_2}
T_2(x,y,z,E) &\leq
\frac{\beta^2}{4}\sum_{\substack{u,v,s,t\in \Tr\\\{u,v\}\neq \{s,t\}}} \1\{\{u,v\},\{s,t\}\in E\}J_{u,v}J_{s,t}
\langle \sigma_0\sigma_x\rangle 	\nl 	
& \qquad \times
	\mathbf E^{0x,\varnothing} \bigg[
	\1\{ u \text{ or }v \in \mathbf C(0)\}
	\1\{ s \text{ or }t \in \mathbf C(0)\}
	\int_0^1 \theta \, \langle \sigma_y\sigma_z;\sigma_u\sigma_v;\sigma_s\sigma_t\rangle_{\theta}
\, \mathrm{d}\theta
\bigg]
.
\end{align}

It is convenient to relabel $u = u_1$, $v = u_2$, $t=t_1$, $s = t_2$.
For the integral over $\theta$,
we use a minor modification\footnote{In \cite[Proposition~5.2]{AF86}, the Ising model is studied with an external magnetic field $h$. For the random current representation, this involves the introduction of a \emph{ghost} vertex $g$
with spin $\sigma_g=1$, and \cite[Proposition~5.2]{AF86} gives a version of
\eqref{eq: triple truncated bound} in which the vertex $z$ is equal to the ghost.
The proof of \cite[Proposition~5.2]{AF86} adapts to our setting with minor modifications.}
of \cite[Proposition~5.2]{AF86},
followed by
the monotonicity of the two-point function in $\theta$, to get
\begin{equation} \label{eq: triple truncated bound}
\langle \sigma_y \sigma_z; \sigma_{u_1} \sigma_{u_2}; \sigma_{t_1} \sigma_{t_2} \rangle_{\theta}
\le
\sum_{i,j = 1}^2
2 \langle \sigma_y \sigma_{u_i} \rangle
\langle \sigma_{u_{\rho(i)}} \sigma_{t_j} \rangle
\langle \sigma_{t_{\rho(j)}} \sigma_z \rangle
+ (y \Leftrightarrow z ) ,
\end{equation}
where $\rho(1) = 2$, $\rho(2) = 1$,
and $(y\Leftrightarrow z)$ denotes the sum with the roles of $y$ and $z$ interchanged.
This allows us to bound the integral using $\int_0^1 \theta \mathrm d\theta = 1/2$.
For the remaining expectation in \eqref{eq: simplification T_2},
we use \cite[Proposition~A.3]{AD21},
which leads to
\begin{equation} \label{eq: multiconnectivity proba}
\langle \sigma_0\sigma_x\rangle \mathbf P^{0x,\varnothing}[u_1 \textup{ or }u_2 \in \mathbf C(0), t_1\textup{ or }t_2 \in \mathbf C(0)]
\leq
\sum_{i,j=1}^2  \langle \sigma_0\sigma_{u_i} \rangle\langle \sigma_{u_i} \sigma_{t_j} \rangle\langle \sigma_{t_j}\sigma_x\rangle
+ (0\Leftrightarrow x),
\end{equation}
where $(0\Leftrightarrow x)$ denotes the sum
with the roles of $0$ and $x$ interchanged.

We insert the above two inequalities into \eqref{eq: simplification T_2},
average over $E=\omega$,
and then sum over $x,y,z\in \Tr$.
The sums over $x,y,z$ produce $\chi^3$.
Then, by translation invariance,
we can shift the origin to one of $u_1,u_2,t_1,t_2$, and then sum over the former origin to get another factor of $\chi$.
The remaining sums are of three kinds and can all be written as convolutions.
In terms of the torus convolution $\star$ defined in \eqref{eq:starconv}, this produces
\begin{align}
\label{eq:T2_finalbd}
&\sum_{x,y,z\in \mathbb T_r}\mathbb E_p[T_2(x,y,z,\omega)]
\nonumber \\  & \qquad
\le
\frac{\beta^2 p^2} 4 (\half) (2)(4) \chi^4
	\Big( 4 \norm{ J \star \tau^2 \star J }_1
	+ 8 \norm J_1 (\tau \star J \star \tau)(0)
	+ 4 (J \star \tau \star J \star \tau)(0)
	\Big) ,
\end{align}
where the factor 2 is from \eqref{eq: triple truncated bound}
and the first factor 4 is from $(0\Leftrightarrow x)$ and $(y\Leftrightarrow z)$.
We then apply the Cauchy--Schwarz
inequality to see that
\begin{align}
\norm{ J \star \tau^2 \star J }_1
&
=
\norm J_1 ^2 \norm{ \tau^2 }_1
,
\\
\norm J_1 (\tau \star J \star \tau)(0)
&\le \norm J_1 ^2 \norm{ \tau }_2^2
 ,
\\
(J \star \tau \star J \star \tau)(0)
&\le \norm J_1 ^2  \norm{ \tau }_2^2
 .
\end{align}
These right-hand sides are all equal to $|J|^2 \bubble(\beta)$.
Altogether, this
gives an upper bound
$16 \beta^2 p^2 \chi^4 \abs J^2 \bubble$ for the right-hand side of \eqref{eq:T2_finalbd},
which concludes the proof.
 \end{proof}

\section{Percolation}
\label{sec:perc}

The conjectures for the Ising model in
Section~\ref{sec:conjectures} have counterparts
for Bernoulli bond (or site) percolation
on $\Z^d$ in dimensions $d \ge 6$, which we discuss now.

Let $W^*$ denote the Brownian excursion of length $1$,
let $\Psi(x) = \mathbb{E} \exp[x\int_0^1 W^*(t)\mathrm dt] $ denote the
moment generating function of the Brownian excursion area, let
\begin{align}
F(x,s) = \frac 1 6 x^3 - \frac s 2 x^2 + \frac{ s^2} 2 x
,
\end{align}
and define the conjectured \emph{percolation profile}
\begin{equation}
f_{\rm perc}(s) = \int_0^\infty x^2 \mathrm d\sigma_s, \qquad
\mathrm d \sigma_s =
	\frac 1 { \sqrt{2\pi} } x^{-5/2} \Psi(x^{3/2}) e^{-F(x, s)} \mathrm dx
    .
\end{equation}
On the torus,
we write $\tau^{\T_r}_p(0,x) = \mathbb{P}^\Tr_p(0 \connect{} x)$ for the
percolation two-point function, and
$\chi^{\T_r}(p) = \sum_{x\in \T_r}\tau^{\T_r}_p(0,x)$ for
the susceptibility.
On the infinite lattice $\Z^d$, we denote the two-point function by $\tau_p$
and the critical value by $p_c$.

\begin{Conj}[Universal profile for percolation]
\label{conj:perc}
For $d > 6$, with \emph{window scale} $w_r =a_dr^{-d/3}$
for suitably chosen $a_d>0$,
and for $s \in \R$, as $r \to \infty$,
\begin{align}
\label{eq:perc-chi}
    \chi^{\T_r}(p_c + sw_r) &\sim {\rm const}_d \, f_{\rm perc}(s)
        r^{d/3},
    \\
\label{eq:perc-tau}
    \tau^{\T_r}_{p_c + sw_r}(0,x)  - \tau_{p_c}(0,x)
    &\sim
    {\rm const}_d\, f_{\rm perc}(s)
        r^{-2d/3}.
\end{align}
\end{Conj}

We also conjecture that
\eqref{eq:perc-chi} and \eqref{eq:perc-tau} hold for percolation
under free boundary conditions (i.e., only using bonds in $\Lambda_r$),
after $p_c$ is replaced by an effective critical value $p_{c,r}=p_c+ v_r$
with $v_r \asymp r^{-2}$.  If this is true, then since $v_r$ is larger
than the window $w_r$, the
two FBC and PBC windows
would not overlap.

Upper and lower bounds consistent with \eqref{eq:perc-chi} and \eqref{eq:perc-tau} are proved in \cite{HMS23} for spread-out percolation in dimensions $d>6$,
but asymptotic formulas have not been proved and the role of the profile has not been demonstrated.
For $d=6$, nothing close to Conjecture~\ref{conj:perc} has been proved, though we expect that the profile still applies after appropriate
logarithmic corrections to $w_r$, $v_r$, and to the factors $r^{d/3}$
and $r^{-2d/3}$ in \eqref{eq:perc-chi}--\eqref{eq:perc-tau}.

Logarithmic corrections for $d=6$ are discussed in the physics literature,
e.g., in \cite{Ahar80,Ruiz98}, on the basis of connections with the
$q \to 1$ limit of the $q$-state Potts model and $\varphi^3$ field theory.
Mathematically, the connection between percolation and $\varphi^3$ theory remains
mysterious, and to our knowledge there are no rigorous results in this direction.
Indeed, the very definition of the $\varphi^3$ model is problematic, since
$e^{\varphi^3}$ is not integrable.
Nevertheless, there are indications of a $\varphi^3$ structure in the mathematical
work on percolation:
the percolation diagram in Figure~\ref{figure:diagrams other models} has vertices
of degree 3, as do the tree-graph diagrams of \cite{AN84}, and the
lace
expansion diagrams of \cite{HS90a}.
Also, the $x^3$ in the function
$F$ in $ \mathrm d\sigma_s$ is reminiscent of $\varphi^3$.

The conjectured profile arises by analogy with the Erd\H{o}s--R\'{e}nyi random graph, where the properly rescaled cluster size (without taking expectation) is known to converge in distribution to a random variable described by the Brownian excursion \cite{Aldo97}.
The limiting random variable is also characterised by a point process \cite{JS07}.
The measure $\sigma_s$ is the intensity of that point process and is taken from \cite[Theorem~4.1]{JS07}.
The point process describes cluster sizes, in the sense that
\begin{equation}
n^{-2k/3} \sum_i \abs{ C_i }^k
\Rightarrow
\int_0^\infty x^k \mathrm d\sigma_s
\qquad
(k \ge 2)
\end{equation}
in distribution \cite{Aldo97}.
The $k=2$ case gives the properly rescaled susceptibility and identifies $f_{\rm perc}(s)$ in \eqref{eq:perc-chi}.
To our knowledge, the convergence of the rescaled susceptibility has not been proved, not even for the Erd\H{o}s--R\'{e}nyi random graph.

\section*{Acknowledgements}
We thank Hugo Duminil-Copin, Trishen Gunaratnam, and Alexis Prévost for useful discussions. We also thank two anonymous referees for useful comments.
GS thanks the Section de Math\'ematiques of the University of Geneva for kind
hospitality during a visit when this work was initiated.
The work of YL and GS was supported in part by NSERC of Canada. The work of RP was supported by the Swiss National Science Foundation, the NCCR SwissMAP, and the European Research Council (ERC) under the European Union’s Horizon 2020 research and innovation programme (grant agreement No. 757296).


\end{document}